\documentclass[12pt]{article}

\usepackage{a4wide}

\usepackage{amssymb}

\usepackage{pj-arx-dpda-macros}

\title{A Short Decidability Proof for DPDA Language Equivalence 
via First-Order Grammars
}

\author{Petr Jan\v{c}ar\\
{\small Techn. Univ. Ostrava, Czech Republic,}\\
{\small http://www.cs.vsb.cz/jancar/,
 email:petr.jancar@vsb.cz}
}

\date{}

\begin{document}

\maketitle

\begin{abstract}
\noindent
The main aim of the paper is to give a short 
self-contained proof of the decidability of
language equivalence for deterministic pushdown automata, 
which is the famous problem 
solved by G. S\'enizergues, for which  C. Stirling has 
derived a primitive recursive complexity upper bound.
The proof here is given in 
the framework of first-order grammars,
which seems to be particularly apt for the aim.
\end{abstract}


\section{Introduction}
\label{sec:intro}

The decidability question for language equivalence of two
deterministic pushdown automata (dpda) is  a famous problem in language
theory. The question was explicitly stated in the 1960s 
\cite{Ginsburg66} (when language inclusion was found undecidable);
then a series of works solving various subcases followed,
until the question was answered positively
by S\'enizergues in 1997 
(a full version appeared in~\cite{Senizergues:TCS2001}).
G. S\'enizergues was
awarded G\"odel prize in 2002 for this significant achievement. 
Later Stirling~\cite{Stirling:TCS2001},
and also S\'enizergues~\cite{Senizergues:TCS2001_simple},
provided simpler proofs than the original long
technical proof. A modified version, which 
showed a primitive recursive complexity upper bound, appeared as a conference
paper by Stirling in 2002~\cite{Stir:DPDA:prim}.

Nevertheless, even the above mentioned simplified proofs are rather
technical, and they seem not well understood in the (theoretical)
computer science community.
One reason might be that the frameworks like that of strict deterministic
grammars, which were chosen by S\'enizergues and Stirling, are not
ideal for presenting this topic to a broader audience. 

From some (older) works by Courcelle,
Harrison and others we know that the dpda-framework
and strict-deterministic grammar framework are equivalent 
 to the
framework of first-order schemes, or first-order grammars. 
In this paper, a proof in the framework of first-order grammars is
presented.

\emph{Author's remark.}
I have been reminded that J.~R.~B\"uchi in his book 
``Finite automata, their algebras and grammars: towards a theory of
formal expressions'' (1989)
argues that using terms
is the way proofs on context-free grammars should be done. 
I have not managed to verify myself, but this can be an indication that
the framework of first-order terms might be ``inherently more
suitable'' here.

In fact, the proof here shows
the decidability of trace equivalence (a variant of
language equivalence, coinciding with  bisimulation equivalence on 
deterministic labelled transition systems) for deterministic first-order
grammars; the states (configurations) are first-order terms which can
change by performing actions according to the root-rewriting rules.
To make the paper self-contained, a reduction from 
the dpda language equivalence problem to the above trace equivalence
problem is given in an appendix.

In principle, the proof is lead in a similar manner as the proofs of
S\'enizergues and Stirling, being based on the same abstract ideas.
Nevertheless, the framework of the first-order terms seems to allow to
highlight the basic ideas in a more
clear way and to provide
a shorter proof in the form of
a sequence of relatively simple observations.
Though the proof is the ``same'' as the previous ones on an abstract
level, and each particular idea used here might be embedded somewhere in the
previous proofs, it is by no means a ``mechanical translation'' of a
proof (or proofs) from one framework to another.
Another appendix then shows that the framework
chosen here also has a potential to concisely comprise
and slightly strengthen the previous
knowledge of the complexity of the problem, though the proofs in that
part are not so detailed as in the decidability part.

\emph{Author's remarks.}
I hope that the presentation here should significantly extend the number of
people in the community who will 
understand
the problem which seems to belong to fundamental ones; this might also
trigger new attempts regarding the research of complexity. 
Another remark concerns
the previous version(s) of this arxiv-paper where I also claimed to
provide a smooth generalization of the decidability proof to the case
of bisimilarity for nondeterministic first-order grammars.
G\'eraud S\'enizergues was present at my talk at 
http://www.lsv.ens-cachan.fr/Events/Pavas/ (20 January, 2011) and he
later put a counterexample on arxiv: http://arxiv.org/abs/1101.5046.
(My mistake was, in fact, embarrassingly simple, mixing the
absolute equivalence levels with the eq-levels relative to fixed
strategies. At the moment, I do not speculate how this can be
corrected.)

In the rest of this section, the main ideas are sketched.
A GNF (Greibach Normal Form) grammar $\calG=(\calN,\act,\calR)$, with finite sets
$\calN$ of nonterminals and $\act$ of terminals, 
has the rewriting rules 
$X\gt{}a Y_1\dots Y_n$ where $a\in\act$ and $X,Y_1,\dots,Y_n\in\calN$.
Such a rewrite rule can be written as  $Xx\gt{a} Y_1\dots Y_nx$, for a
formal variable $x$, and read as follows: 
any sequence $X\alpha\in\calN^*$ (a `state', or a `configuration') can perform
action $a$ while changing into  $Y_1\dots Y_n\alpha$\,;
this includes the case  when $n=0$ and thus $Y_1\dots Y_n=\varepsilon$, the
empty word.
The language $L(\alpha)$ is the set of words $w=a_1a_2\dots a_m\in\act^*$ such that 
$\alpha\gt{a_1}..\gt{a_2}..\cdots ..\gt{a_m}\varepsilon$. 

In a first-order grammar $\calG=(\calN,\act,\calR)$, each nonterminal $X$
has a finite arity (not only arity 1 like in the GNF grammar), 
and the (root rewriting) rules are $Xx_1\dots x_m\gt{a}E(x_1,\dots,x_m)$
where $m=arity(X)$ and $E$ is a finite term over $\calN$ where all
occurring variables are from the set $\{x_1,\dots,x_m\}$.
When $G_1,\dots,G_m$ are terms then 
 $XG_1\dots G_m\gt{a}E(G_1,\dots,G_m)$ where $E(G_1,\dots,G_m)$
 denotes the result of substitution
$E(x_1,\dots,x_m)[G_1/x_1,\dots,G_m/x_m]$.
Grammar $\calG$ is \emph{deterministic} if for each pair $X\in\calN$,
$a\in\act$ there is at most one rule of the type $Xx_1\dots x_m\gt{a}\dots$.
We note that states (configurations) are no longer strings (as in the
case of GNF grammars) but
\emph{terms}, naturally represented as \emph{trees}. 

It is a routine to reduce dpda language equivalence 
to deterministic first-order grammar \emph{trace equivalence}:
two terms $T,T'$ are equivalent, $T\sim T'$,
iff the words (traces) $w\in\act^*$ enabled in $T$
($T\gt{w}T_1$ for some $T_1$) are the same as those enabled in $T'$.
It is natural to define the equivalence-level
$\eqlevel(T,T')$  as the maximal $k\in\Nat$ for which we have 
$T\sim_k T'$, which means that 
$T,T'$ enable the same words upto length $k$; 
we put $\eqlevel(T,T')=\omega$ iff $T\sim T'$, i.e. iff  
$T\sim_k T'$ for all $k\in\Nat$.

Given a deterministic first-order grammar
$\calG$ and an initial pair of terms $T_0,U_0$,  the first idea for
deciding $T_0\stackrel{?}{\sim} U_0$ 
is to use the (breadth-first) search for a
shortest 
word $w$ which is enabled in just one of $T_0,U_0$.
We call such a
\emph{word} as \emph{offending} (adopting the viewpoint of a
defender of the claim $T_0\sim U_0$).

To get a
terminating algorithm in the case of $T_0\sim U_0$,
we can think of some \emph{sound system} enabling to establish for certain words
$u\in\calA^*$ that they are not \emph{offending prefixes}, i.e.
prefixes of offending words for  $T_0,U_0$; e.g., if $T_0\gt{u} T$
and $U_0\gt{u} T$ then $u$ cannot be an offending prefix.
We aim at \emph{completeness},
i.e., look for some means which enable to recognize sufficiently many
nonoffending prefixes, finally showing that $\varepsilon$ is not
offending, which means  $T_0\sim U_0$.

We use a simple observation that
the eq-level of a pair $(T,U)$ can drop by at most $1$ when both sides
perform the same
action, 
and it 
really \emph{drops by $1$ in each step} when we follow an offending
word.
It is also easy to observe a \emph{congruence property} of
\emph{subterm replacement}; in particular, given terms $U_1,U_2$ with 
$\eqlevel(U_1,U_2)=k$ and $T_1,T_2$ with  $\eqlevel(T_1,T_2)\geq k+1$
where $U_1$ has a subterm $T_1$, i.e. $U_1=E(T_1)$,
by replacing $T_1$ with $T_2$ we get for the arising $U'_1=E(T_2)$ 
that $\eqlevel(U'_1,U_2)=\eqlevel(U_1,U_2)=k$; 
the eq-level has been unaffected, and moreover, even the offending words for
$(U'_1,U_2)$ 
are
the same as the offending words for $(U_1,U_2)$. 

The above simple observations allow to build a (sound and) 
complete system,
when we add the notion of a \emph{basis}, a finite set of pairs of `equivalent heads'
(tree-tops, tree-prefixes)
$E(x_1,\dots,x_n), F(x_1,\dots,x_n)$, for which 
we have $E(T_1,\dots,T_n)\sim F(T_1,\dots,T_n)$ for every instance.
To enable a smooth completeness proof, showing
even the existence of a fixed sufficient basis
$\calB$ for each grammar $\calG$ (not depending on the initial pair),
it is helpful
to start immediately in a more general setting of
\emph{regular terms}, 
which are finite or infinite terms with only finitely many
subterms (where a subterm can possibly have an infinite number of
occurrences); such terms have natural finite graph presentations.

The structure of the paper is clear from (sub)section titles.
There are also two (above mentioned) appendices.

\section{Basic definitions and simple facts}
\label{sec:definitions}

In this section we introduce the basic definitions and observe 
some simple facts on which the main proof is based. 
\\
By $\Nat$ we denote the 
set $\{0,1,2,\dots\}$ of natural numbers; symbol $\omega$ 
is taken as the first infinite ordinal number,
which satisfies $n<\omega$ and $\omega-n=\omega$ for any
$n\in \Nat$.
\\
For a set $A$, by $A^*$ 
we denote the set of finite sequences, i.e. words, of
elements of $A$; the length of $w\in A^*$ is denoted 
$|w|$,
and we use $\varepsilon$ for the empty sequence,
so $|\varepsilon|=0$.
By $A^{\omega}$ we denote the set of infinite
sequences of
elements of $A$ (i.e., the mappings $\Nat\rightarrow A$).
\\
Given 
$L\subseteq A^*\cup A^{\omega}$
and $u\in\act^*$, by $u\lquot L$ we mean the
\emph{left quotient} of $L$ by $u$, i.e. the set
$\{v\in\calA^*\cup\calA^{\omega}\mid uv\in L\}$. 
By $\pref(L)=\{u\in\act^* \mid uv\in L$ for some
$v\}$ we denote the set of (finite) prefixes of the words in $L$.

\subsubsection*{(Deterministic) labelled transition systems and 
(stratified) trace
equivalence}

A \emph{labelled transition system (LTS)} is a tuple
$(\calS,\calA,(\gt{a})_{a\in\calA} )$
where $\calS$ is the set of \emph{states}, $\calA$ the set of
\emph{actions} 
and $\gt{a}\subseteq \calS\times\calS$ is the set of
\emph{transitions labelled with} $a$.
\\
We extend $\gt{a}$ to relations $\gt{w}\subseteq
\calS\times \calS$ for all $w\in\calA^*$ inductively:
$s\gt{\varepsilon}s$; if $s\gt{a}s_1$ and $s_1\gt{v}s_2$ then
$s\gt{av}s_2$. We say that $s'$ is \emph{reachable 
from} $s$ \emph{by
a word} $w$ if $s\gt{w}s'$.
A \emph{state} $s\in\calS$ \emph{enables (a trace)}
$w\in\act^*$, denoted $s\gt{w}$, if there is $s'$ such that 
 $s\gt{w}s'$. An infinite sequence 
 $\alpha=a_1a_2a_3\dots\in\act^{\omega}$ 
is enabled in $s_0$, written $s_0\gt{\alpha}$, if there are
$s_1,s_2,s_3,\dots$ such that $s_i\gt{a_{i+1}}s_{i+1}$ 
for all $i\in\Nat$.
The \emph{trace-set} of a state $s\in \calS$ is defined as
\begin{center}
$\traces(s)=\{w\in\act^*\mid s\gt{w}\}$. 
\end{center}
For $k\in\Nat$, 
we define
$\traces^{\leq k}(s)=\{w\in\act^*\mid w\in\traces(s)$ and $|w|\leq k\}$;
we also put
$\omtraces(s)=\{\alpha\in\act^{\omega}\mid s\gt{\alpha}\}$.
On the set $\calS$, we define the \emph{(trace) equivalence} $\sim$, and the \emph{family  
of equivalences} $\sim_k$, $k\in\Nat$, as follows:
\begin{center}
$r\sim s$ $\defined$ $\traces(r)=\traces(s)$ 
\ and 
\ $r\sim_k s$ $\defined$ $\traces^{\leq k}(r)=\traces^{\leq k}(s)$.
\end{center}

\begin{observ}\label{prop:basicstratif}
$\sim_0=\calS\times\calS$.
$\forall k\in\Nat: \sim_k\supseteq \sim_{k+1}$.
$\cap_{k\in\Nat}\sim_k=\sim$.

\end{observ}
The \emph{equivalence level}, or the \emph{eq-level},
of a pair of states is defined as follows:
\begin{quote}
$\eqlevel(r,s)=k$ if $r\sim_k s$ and $r\not\sim_{k+1} s$\,;
$\eqlevel(r,s)=\omega$ if $r\sim s$.
\end{quote}
The shortest words 
showing nonequivalence for a pair $r,s$ (if $r\not\sim s$)
are called \emph{offending words}:
$\offpl(r,s)=\{w\mid w$ is a shortest word in 
$(\traces(r)\smallsetminus\traces(s))\cup
(\traces(s)\smallsetminus\traces(r))\}$. 
The elements of $\pref(\offpl(r,s))$ are called \emph{offending
prefixes} for the pair $(r,s)$.
We now note some trivial facts, point $3$ being of particular 
interest:

\begin{observ}\label{prop:basicoffwords}
(1.) $r\sim s$ iff $\offpl(r,s)=\emptyset$ iff $\varepsilon\not\in \pref(\offpl(r,s))$.
\\
(2.) If $r\not\sim s$ then $\eqlevel(r,s)=|w|{-}1$ for any
$w\in\offpl(r,s)$.
\\
(3.) If $\eqlevel(r,s)=k$ and $\eqlevel(r,q)\geq k+1$ then 
\\
\hspace*{3em}$\eqlevel(q,s)=k$ and $\offpl(r,s)=\offpl(q,s)$. 
\end{observ}
An \emph{LTS}  $(\calS,\calA,(\gt{a})_{a\in\calA})$ is 
\emph{deterministic} 
if each $\gt{a}$ is a partial
function, i.e.:
if $r\gt{a}s_1$ and $r\gt{a}s_2$ then
$s_1=s_2$. (Recall now the left quotient operation $u\lquot L$.)

\begin{observ}\label{prop:basicquottraces}
In any \emph{deterministic} LTS,
if $r\gt{u} r'$ then 
$\traces(r')=u\lquot \traces(r)$, and
$\omtraces(r')=u\lquot \omtraces(r)$.
Moreover, $\traces(r)=\traces(s)$ implies  
$\omtraces(r)=\omtraces(s)$.
\end{observ}
We use notation $(r,s)\gt{u} (r',s')$ as a shorthand meaning
$r\gt{u} r'$ and $s\gt{u} s'$. 

\begin{prop}\label{prop:basiceqleveldrop}
Assume a \emph{deterministic} LTS, and suppose
$(r,s)\gt{u} (r',s')$. Then:
\\
1. $\eqlevel(r,s)-|u|\leq \eqlevel(r',s')$ 
(in particular, $r\sim s$ implies
$r'\sim s'$);
\\
2. if $r\not\sim s$ then 
$\eqlevel(r,s)-|u|= \eqlevel(r',s')$ iff 
$u\in\pref(\offpl(r,s))$; 
\\
3. if $u\in\pref(\offpl(r,s))$ then 
$\offpl(r',s')=u\lquot\offpl(r,s)$.
\end{prop}

\begin{proof}
This follows almost trivially from 
Observation~\ref{prop:basicquottraces}. 
E.g., for Point 3. it is sufficient to note:
if $u\in\pref(\offpl(r,s))$ then 
$uv\in \offpl(r,s)$ iff
$v\in\offpl(r',s')$.  
\qed
\end{proof}

\subsubsection*{Finite and infinite regular terms and their
finite graph presentations} 

We now give (a variant of) standard definitions of first-order terms,
including infinite terms;
we fix a countable set 
$\calV=\{x_1,x_2,x_3,\dots\}$ of (first-order) \emph{variables}.

Let us now assume  a given \emph{finite} set $\calN$ of 
ranked symbols, called \emph{nonterminals}
(though we can also view them as function symbols).
Each $X\in \calN$ thus has  $\arity(X)\in\Nat$; we use $X,Y$ to range
over elements of $\calN$.

A (general) \emph{term} $E$ over $\calN$ (and $\calV$) is defined as a
partial mapping $E:\Nat^*\rightarrow \calN\cup\var$
\\
where  
the domain $\dom(E)\subseteq \Nat^*$ is prefix-closed,
i.e. $\dom(E)=\pref(\dom(E))$, and nonempty ($\varepsilon\in\dom(E)$);
moreover, for $\gamma\in\dom(E)$ we have 
$\gamma i\in \dom(E)$ iff $1\leq i\leq
\arity(E(\gamma))$ where the arity of variables $x_j\in\var$ is viewed as
$0$.

For each $\gamma\in\dom(E)$, by $E_{[\gamma]}$ we denote
\emph{the subterm occurring} at $\gamma$ in $E$ where
$E_{[\gamma]}(\delta)=E(\gamma\delta)$ for each 
$\delta\in\dom(E_{[\gamma]})=\gamma\lquot \dom(E)$;
this \emph{occurrence of subterm} $E_{[\gamma]}$ has \emph{depth}
$|\gamma|$ in $E$. 

For $X\in\calN$ and terms $G_1,G_2,\dots, G_m$, where 
$m=\arity(X)$, by $XG_1G_2\dots G_m$ we denote the term $E$ for which
$E(\varepsilon)=X$ and $E_{[i]}=G_i$ for each length-1 sequence $i$
where $1\leq i\leq m$; $X$ is the \emph{root nonterminal} of this term
$E$. Each variable $x_j\in\var$ is also viewed as
the term $E$ for which $E(\varepsilon)=x_j$ 
(and thus $\dom(E)=\{\varepsilon\}$). 

A \emph{term} $E$ is \emph{finite} (\emph{infinite})
if $\dom(E)$ is finite (infinite).
The \emph{depth-size} of a finite term $E$, denoted $\depthsize(E)$,
is the maximal $|\gamma|$ for $\gamma\in\dom(E)$
(i.e., the maximal depth of a subterm-occurrence in $E$).
A \emph{term} is \emph{regular} if the set of its subterms
is finite 
(though the subterms can have infinitely many occurrences).

By $\trees_{\calN}$ we denote the set of all
(finite and infinite) \emph{regular} terms, since we will not consider
nonregular terms anymore.
Hence from now on, when saying ``term'' we mean ``regular term''. 
$\gtrees_{\calN}$ denotes the set of all (regular)
\emph{ground terms}, i.e. the terms  in which no  variables $x_i$
occur.
We use symbols $T,U,V,W$ (possibly with sub- and
superscripts) for ranging over $\gtrees_{\calN}$;
symbols $E,F,G,H$ 
are used more generally, they range over $\trees_{\calN}$.

Regular terms can be infinite but they have natural finite
presentations, since they can be viewed as the unfoldings of finite
graphs: 

\begin{defn}
A \emph{graph presentation} of a regular term is 
a finite labelled (multi)graph $g$, where each node $v$ has a label
$\lambda(v)\in\calN\cup\var$ and  $m$ outgoing edges
labelled with $1,2,\dots,m$ where $m=\arity(\lambda(v))$ 
(and where different edges can have the same target); 
moreover, one node is selected  as the root. 
Graph $g$ represents term $\calT_g$ as follows:  $\dom(\calT_g)$
consists of sequences of edge-labels of
finite paths in $g$  starting in the root; 
$\calT_g(\gamma)=\lambda(v')$
where $v'$ is the end-vertex of the path with the edge-label sequence
$\gamma$.
\end{defn}
(Given $\calN$,) 
we can naturally define a notion of the \emph{size of a graph
presentation} $g$, e.g., as the number of nodes of $g$, or, to be
pedantic, as the length of a standard bit-string representation of $g$
(thus also handling the descriptions of indexes of variables $x_i$).
We define the
\emph{presentation size} of a term $F$, denoted 
$\pressize(F)$, as the size of the least graph
presentation of $F$. 
By $\pressize(E,F)$ for a pair $E,F$ we mean the sum 
$\pressize(E)+\pressize(F)$, say. 
On our level of reasoning, we do not need further technical details,
since the facts like the following one are sufficient for us.

\begin{observ}\label{prop:boundpressize}\hfill
\\
For any $s\in\Nat$, there are only finitely many pairs $(E,F)$ with
$\pressize(E,F)\leq s$.
\end{observ}
We will also use some facts concerning an effective (algorithmic)
work with finite
presentations of regular terms. The next observation is an example of
such a fact. Further we often leave such facts implicit.

\begin{observ}\label{prop:algorgraphequality}
There is an algorithm which, given ($\calN$ and) 
graph presentations $g_1$, $g_2$, decides if $\calT(g_1)=\calT(g_2)$.
\end{observ}

\subsubsection*{Substitutions, 
ground instances of a pair
$(E,F)$, a limit substitution}

By a \emph{substitution}
we mean a mapping $\sigma:\var\rightarrow \trees_{\calN}$.
The term $E\sigma$ arises from $E$ by 
replacing each occurrence $x_i$ in $E$ with $\sigma(x_i)$.
(Since $E$ and $\sigma(x_i)$ are regular, $E\sigma$ is regular.)
By writing $\sigma=[G_{i_1}/x_{i_1},\dots, G_{i_n}/x_{i_n}]$
we mean that $\sigma(x_i)=G_i$ if $i\in\{i_1,\dots,i_n\}$ and 
$\sigma(x_i)=x_i$ otherwise. Substitutions can be naturally composed;
associativity allows to omit the parentheses:
$E\sigma_1\sigma_2\sigma_3=((E\sigma_1)\sigma_2)\sigma_3=
E(\sigma_1(\sigma_2\sigma_3))=(E\sigma_1)(\sigma_2\sigma_3)$, etc.
\\
$F$ is an \emph{instance} of $E$ if there is 
a substitution $\sigma$ such that $E\sigma=F$. 

As usual, we sometimes write $E(x_{i_1},\dots,x_{i_n})$
to denote the fact
that all variables occurring in $E$  
are from the set
$\{x_{i_1},\dots,x_{i_n}\}$.
In fact, we only use the special case $E(x_{1},\dots,x_{n})$;
note that even a ground term $E$ 
can be viewed
as $E(x_1,\dots,x_n)$, for any $n$.
(We ignore the slight notational collision with 
the previous use $E(\gamma)$, since this should cause no problems.)

\smallskip
\noindent
\textbf{Convention.} When writing $F(G_1,\dots,G_n)$, we implicitly
assume $F=F(x_1,\dots,x_n)$, and we take 
 $F(G_1,\dots,G_n)$ as a
shorthand for $F[G_1/x_1,\dots, G_n/x_n]$. In particular
we note that $F(G_1(H_1,\dots,H_m),\dots,G_n(H_1,\dots,H_m))=
(F(G_1,\dots,G_n))(H_1,\dots,H_m)$.

\smallskip

\noindent
A ground term $U$ which is an instance of $E$ is called
a \emph{ground instance} of $E$.
We will in particular use the notion
of a \emph{ground instance of a pair} 
$(E(x_1,\dots,x_n),F(x_1,\dots,x_n))$: it is any 
pair $(E\sigma,F\sigma)$
where $\sigma$ is a substitution $[U_1/x_1,\dots, U_n/x_n]$ 
where $U_i$ are ground terms.
We usually write $(E(U_1,\dots,U_n),F(U_1,\dots,U_n))$ instead of 
 $(E\sigma,F\sigma)$.

By writing $E\sigma^1$ we mean $E\sigma$;  $E\sigma^{k+1}$
($k\in\Nat$) means $E\sigma^k\sigma$.
We need just a  special case of substitutions
$\sigma$ for which $E\sigma^{\omega}$ 
is well-defined; we use the graph presentations for the definition
(which also shows another aspect of the effective work with these
finite presentations).

\begin{defn}
Given (a regular term) $H$ we define $H[H/x_i]^{\omega}$, denoted as
$H\limtreei$, as follows: given a graph presentation $g$ where
$\calT(g)=H$, then $H\limtreei=\calT(g')$ where $g'$ arises from $g$
by redirecting each edge leading to a node labelled
with $x_i$ (in $g$) to the root (in $g'$).
(Hence if $H=x_i$ or $x_i$
does not occur in $H$ then 
$H\limtreei=H$.)
\end{defn}

\subsubsection*{Head-tails presentations of terms,
the $d$-prefix form of terms}

If $F=E(G_1,\dots,G_n)$ then we say that 
the (regular) \emph{head} $E$ and the (regular)
\emph{tails}  $G_1,\dots,G_n$ constitute a 
\emph{head-tails presentation} of $F$. 
We can also note that
the head $E$ itself can be presented by a head-tails
presentation $E=G(F_1,\dots,F_m)$, say, etc. 
\\
A particular head-tails presentation of a term is its 
$d$-prefix form; it suffices when we restrict ourselves to ground terms:

\begin{defn}\label{def:dprefix}
For a ground term $V$ and $d\in\Nat$, the
\emph{$d$-prefix
form of} $V$ arises as follows:
we take all (ordered) occurrences of subterms of $V$ with depth $d$ (if
any), say $T_1,\dots,T_n$, 
and replace them with variables $x_1,\dots,x_n$,
respectively. We thus get a \emph{finite} term
$P^V_d=P^V_d(x_1,\dots,x_n)$, 
the \emph{$d$-prefix of}
$V$. 
The head $P^V_d$ and the tails $T_1,\dots,T_n$ constitute the
\emph{$d$-prefix form} of $V=P^V_d(T_1,\dots,T_n)$. 
($P^V_d=V$ when $V$ is a finite term with $\depthsize(V)<d$.)
\end{defn}
\begin{observ}\label{prop:boundtails}
If $m$ is the maximal arity of nonterminals in $\calN$
then the number $n$ of tails in the
$d$-prefix form is bounded by $m^d$.
\end{observ}

We also note the next obvious fact.

\begin{observ}\label{prop:limdprefform}
If $F=F(x_1)\neq x_1$ then $F\limtreeone$ is a ground term, and for
any $T$ we
have $P^{F\limtreeone}_d=P^{F\sigma^d(T)}_d$ where 
$\sigma=[F/x_1]$.
\end{observ}

\subsubsection*{First-order grammars as generators
of LTSs}

\begin{defn}
A \emph{first-order grammar} is a structure 
$\calG=(\calN,\act,\calR)$ where $\calN$ is a finite set of 
ranked \emph{nonterminals}, 
$\act$ is a finite set of \emph{actions} (or terminals), 
and $\calR$ 
a finite set of \emph{(root rewriting) rules} of the form
\begin{equation}\label{eq:rewrule}
Xx_1x_2\dots x_m\gt{a} E(x_1,x_2,\dots,x_m)
\end{equation}
where $X\in \calN$, $\arity(X)=m$, 
$x_1,x_2,\dots,x_m\in \calV$,
$a\in\act$,
and $E$ is a
\emph{finite}  term over $\calN$ (and $\calV$) 
 in which all occurring
variables are from the set
$\{x_1,x_2,\dots,x_m\}$.
($E=x_i$ is a particular example.)
\emph{Grammar} $\calG$ is \emph{deterministic}, a \emph{det-first-order
grammar},
if for each pair $X\in\calN$, $a\in\act$ 
there is at most one rule of the form~(\ref{eq:rewrule}). 
\end{defn}

\smallskip

\noindent
\emph{Remark.}
Context-free grammars  in Greibach normal form
can be seen as  a special case,
where each nonterminal has arity $1$.
Classical rules like $A\rightarrow
aBC$, $B\rightarrow b$ can be presented as 
$Ax_1\gt{a}BCx_1$,
$Bx_1\gt{b}x_1$.

\begin{center}
We view $\calG=(\calN,\act,\calR)$
as a generator of the 
$LTS_{\calG}=(\trees_{\calN},\act,(\gt{a})_{a\in\act})$
\end{center}
where each (root) rewriting rule
$Xx_1x_2\dots x_m\gt{a} E(x_1,x_2,\dots,x_m)$
is a ``schema''
(a~``template'') of a set of transitions:
for every substitution $\sigma=[G_1/x_1,\dots,G_m/x_m]$
(including $\sigma$ with $\sigma(x_i)=x_i$ for all $i$)
we have 
$(Xx_1\dots x_m)\sigma\gt{a}E(x_1,\dots,x_m)\sigma$, 
i.e.
\begin{center}
$XG_{1} G_{2} \ldots G_{m} \gt{a} E(G_1,G_2,\dots,G_m)$.
\end{center}
\begin{observ}
When $\calG$ is deterministic then  $LTS_{\calG}$ is deterministic.
\end{observ}
Though the main result applies to deterministic grammars, we will also 
note some properties holding in the general (nondeterministic) case.

The notions and notation like $F\gt{w}$ ($w$ is enabled by $F$),
$F \gt{w} F'$ (for words $w \in \calA^{*}$), 
$F \gt{\alpha}$ (for $\alpha \in \calA^{\omega}$),
$\traces(F)$, $\traces^{\leq k}(F)$, $\omtraces(F)$,
trace equivalence $E\sim F$, etc.,
are inherited from $LTS_{\calG}$. 
(Note that the term $x_i$ enables no actions; nevertheless,
it is technically convenient to have also
nonground terms as states in  $LTS_{\calG}$.)
\\
We will be interested in comparing ground terms,
so we will use $T\sim_k U$, $T\sim U$, 
$\eqlevel(T,U)$, $\offpl(T,U)$, $\pref(\offpl(T,U))$
mainly for
$T,U\in\gtrees_{\calN}$.
We note (now explicitly)
another fact about the effective work with graph presentations:

\begin{observ}
Given ($\calG$ and) a graph presentation $g$ (of term $\calT(g)$), we can effectively 
find all rewriting rules in $\calR$ (of type~(\ref{eq:rewrule})) which
can be applied to $\calT(g)$ and for each such rule yielding 
$\calT(g)\gt{a}F$ we can effectively construct some $g'$ such that $\calT(g')=F$. 
\end{observ}

\subsubsection*{Root-locality of action performing, 
words exposing subterm occurrences}

Assuming a given first-order grammar $\calG=(\calN,\act,\calR)$,
we observe some consequences of
the fact that the root nonterminal of $XG_1\dots G_m$ 
determines if $a\in\act$ is enabled, and when a transition is performed
then the subterms $G_1,\dots,G_m$ play no role other than being copied
(or ``lost'')
appropriately.

\begin{observ}\label{prop:basictraceperform}
For  $a\in\act$ and  $u\in\act^*$,  
we have $au\in \traces(XG_1\dots G_m)$ iff 
there is a rule $Xx_1\dots x_m\gt{a}E(x_1,\dots,x_m)$ in
$\calR$ such that $u\in\traces(E)$ or $u=u_1u_2$ where 
$E\gt{u_1}x_i$ and $u_2\in\traces(G_i)$ (for some $i,1\leq i\leq m$).
\end{observ}
We say that 
$w\in\act^*$ \emph{exposes the $i$-th successor} of $X\in\calN$
if  $Xx_1\dots x_m\gt{w}x_i$.

\begin{observ}\label{prop:basicsubtermexpos}
Viewing a (regular) term $E$ as a partial mapping
$E:\Nat^*\rightarrow \calN\cup \calV$, we note that there is $u$ such that 
$E\gt{u}x_i$ iff there is $\gamma\in\dom(E)$ where $E(\gamma)=x_i$ and
for each  prefix $\delta j$ of
$\gamma$ there is some $w$ exposing the $j$-th successor of
$E(\delta)$. 
\end{observ}

\begin{observ}\label{prop:segmenttraces}
$\traces(E(G_1,\dots, G_n))=$
\\
$\traces(E(x_1,\dots, x_n))\cup
\bigcup_{1\leq i\leq n}\{uv\mid E\gt{u}x_i$
 and $v\in\traces(G_i)\}$.
\end{observ}
We now look at some simple algorithmic consequences 
of the above observations.

\begin{prop}\label{prop:computwXi}
There is an algorithm which, given 
$\calG=(\calN,\act,\calR)$, computes (and fixes) a word $w(X,i)$
for each pair $(X,i)$, where
 $X\in\calN$, $1\leq i\leq\arity(X)$, so that 
$w(X,i)$ is a shortest word exposing the $i$-th
successor of $X$ if any such word exists and
$w(X,i)=\varepsilon$ otherwise.
\end{prop}

\begin{proof}
Recalling Observations~\ref{prop:basictraceperform}
and~\ref{prop:basicsubtermexpos},
a brute-force systematic search of all $w(X,i)$ is
sufficient:
we finish when
finding that the remaining pairs $(X,i)$, i.e. those
for which $w(X,i)$ have not been
computed, are mutually dependent, i.e., the existence of an exposing 
 $w(X,i)$ for any of the 
 remaining pairs depends on the existence of exposing
 words for some remaining pairs. 
\qed
\end{proof}

\noindent
For later use we note that we can compute a bound bigger than any
$|w(X,i)|$, say
\begin{equation}\label{eq:Mzero}
M_0=1+\max\{\,|w(X,i)| \,\mid\,  X\in\calN, 1\leq i\leq
\arity(X)\,\}\,.
\end{equation}
We also note the bounded increase of the depth-size of finite terms, 
given by the fact that the right-hand sides of 
the (finitely many) rules~(\ref{eq:rewrule}) in $\calR$ are finite terms.

\begin{observ}\label{prop:binc}
(For $\calG$,) there is a (linear) nondecreasing function
 $\boundinc:\Nat\rightarrow\Nat$ such that:
if \mbox{$F\gt{u}F'$} for a finite term $F$ 
then $\depthsize(F')\leq
\depthsize(F)+\boundinc(|u|)$.
\end{observ}
Generally we cannot provide a similar \emph{lower} 
bound for $\depthsize(F')$, 
since some $x_i$ might not occur in $E$ in (\ref{eq:rewrule}).
But we can 
recall the $d$-prefix form from Definition~\ref{def:dprefix}
(defined for ground terms) and note the following obvious fact.

\begin{observ}\label{prop:prefixchange}
If $V=P^V_d(T_1,\dots,T_n)$ and   $|u|\leq d$ then
$V\gt{u}V'$
iff there is some (finite) $E$  such that 
$P^V_d(x_1,\dots,x_n)\gt{u}E(x_1,\dots,x_n)$ 
and
$V'=E(T_1,\dots,T_n)$ (where $\depthsize(E)\leq d+\boundinc(|u|)$). 
Hence if $P^{U}_d=P^V_d$ then  $U\sim_d V$.
\end{observ}

\subsubsection*{Congruence property of $\sim_k$ and $\sim$}

\begin{prop}\label{prop:congrproperty}
If $T\sim_k T'$ then $E(T)\sim_k E(T')$. 
($T\sim T'$ implies $E(T)\sim E(T')$.)
\end{prop}

\begin{proof}
The claim follows from Observation~\ref{prop:segmenttraces}:
$\traces(E(T))$ consists of traces $w\in\traces(E(x_1))$ 
and of traces of the form
$w=uv$ where $E(x_1)\gt{u}x_1$ and $v\in\traces(T)$.
\qed
\end{proof}

\noindent

\begin{prop}\label{prop:congrlim}
If $T\sim_k F(T)$ where $F\neq x_1$ then 
$T\sim_k F\limtreeone$.
\\
Hence $T\sim F(T)$ implies $T\sim F\limtreeone$.
\end{prop}

\begin{proof}
If $T\sim_k F(T)$
then by repeated use of Proposition~\ref{prop:congrproperty} we get
$T\sim_k F\sigma^k(T)$ where $\sigma=[F/x_1]$. 
By Observation~\ref{prop:limdprefform} we get
$P^{F\sigma^k(T)}_k=P^{F\limtreeone}_k$ and
thus $F\sigma^k(T)\sim_k F\limtreeone$
(by Observation~\ref{prop:prefixchange}), which implies
$T\sim_k F\limtreeone$.
\qed
\end{proof}

\begin{cor}\label{prop:corcongrlimequat}
If $T_i\sim_k H(T_1,\dots,T_n)$ where $H\neq x_i$ then
$T_i\sim_k H\limtreei(T_1,\dots,T_n)$ (with no occurrence of $x_i$ in 
$H\limtreei$). In particular, 
if $T_n\sim_k H(T_1,\dots,T_n)$ where  $H\neq x_n$ then  
$T_n\sim_k H\limtreen(T_1,\dots,T_{n-1})$.
\end{cor}

\subsubsection*{A normal form for first-order grammars}

\begin{observ}\label{prop:noexposeequiv}\hfill
\\
If there is no $u$ such that $E(x_1)\gt{u}x_1$ then 
$E(T)\sim E(T')$ for any $T,T'$.
\end{observ}

\begin{defn}\label{def:grammarnormalform}
$\calG=(\calN,\act,\calR)$ is in \emph{normal form} 
if for each  $X\in\calN$ and $i,1\leq i\leq\arity(X)$  there 
is (a shortest) $w(X,i)$ exposing the $i$-th
successor of $X$.
\end{defn}

\begin{prop}\label{prop:transinnormal}
Any $\calG=(\calN,\act,\calR)$ can be effectively transformed
to  $\calG'=(\calN',\act,\calR')$ in normal form, yielding also an
effective mapping $\transf:\trees_{\calN}\rightarrow\trees_{\calN'}$
such that
$T\sim\transf(T)$ (in the disjoint union of $LTS_{\calG}$ and
$LTS_{\calG'}$). 
\end{prop}

\begin{proof}
For each $Y\in\calN$, by $ES_Y=(i^Y_1, i^Y_2,\dots, i^Y_{k_Y})$ (Exposable
Successors of $Y$)
we denote the subsequence 
of $(1,2,\dots,\arity(Y))$ where $i\in ES_Y$ iff there is $w$
exposing the $i$-th successor of $Y$;
by $Y'$ we denote a fresh nonterminal with
$\arity(Y')=k_Y$.
We put $\transf(x_j)=x_j$ and 
$\transf(YG_1\dots G_m)=Y'\, \transf(G_{i^Y_1})\dots
\transf(G_{i^Y_{k_Y}})$.
Each rule $Xx_1\dots x_m\gt{a} E(x_1,\dots,x_m)$ in $\calR$ is transformed to 
the rule $X'x_1\dots x_{k_X}\gt{a} \transf(E)\sigma$
where $\sigma=[x_1/x_{i^X_1}, \dots, x_{k_X}/x_{i^X_{k_X}}]$; thus
$\calR'$ arises.
This guarantees that $Yx_1\dots x_m\gt{u}x_{i^Y_j}$
iff  $Y'x_1\dots x_{k_Y}\gt{u}x_j$.
The claim now follows from
Proposition~\ref{prop:computwXi} and
Observation~\ref{prop:noexposeequiv}.
\qed
\end{proof}

\subsubsection*
{Exposing equations for pairs $(E,F)$ in deterministic $LTS_{\calG}$}

We now note an important notion 
and make some 
observations.

\begin{defn}
We say that \mbox{$u\in\act^*$} 
\emph{exposes an equation for 
the pair} 
\\
$(E(x_1,\dots,x_n), F(x_1,\dots,x_n))$ 
if 
$E(x_1,\dots,x_n)\gt{u} \modr{x_i}$, 
$F(x_1,\dots,x_n)\gt{u} \modr{H(x_1,\dots,x_n)}$,
or vice versa, where $H(x_1,\dots,x_n)\neq
x_i$ (but might be $H=x_j$ for $i\neq j$).
\\
Formally we write this exposed equation 
as {\modr{$x_i\doteq H(x_1,\dots,x_n)$}}.
\end{defn}

\begin{prop}\label{prop:exposing}
For deterministic $\calG$,
if $E(T'_1,  \dots T'_{n})\sim_k F(T'_1,  \dots, T'_{n})$
and  $E(T_1,  \dots T_{n})\not\sim_k F(T_1,  \dots, T_{n})$
then
an offending prefix for  $(E(T_1,  \dots, T_{n}),F(T_1,  \dots,
T_{n}))$ exposes an equation for $E,F$.
\end{prop}

\begin{proof}
There is surely no 
$u\in\pref(\offpl(E(T_1,  \dots, T_{n}),F(T_1,  \dots, T_{n})))$
such that $(E,F)\gt{u}(x_i,x_i)$. If there is 
$wa\in\offpl(E(T_1,  \dots, T_{n}),F(T_1,  \dots, T_{n}))$ 
(where $|w|<k$)
such that 
$(E,F)\gt{w}(E',F')$ where none of $E',F'$ is a variable and
$E'\not\sim_1 F'$ then 
$wa$ witnesses that
$E(T'_1,  \dots, T'_{n})\not\sim_k F(T'_1,  \dots, T'_{n})$ -- a
contradiction.
\qed
\end{proof}

\noindent
We recall that Observation~\ref{prop:basicquottraces} 
and Proposition~\ref{prop:basiceqleveldrop} apply
to $LTS_{\calG}$  when 
$\calG$ is deterministic, and observe the following:

\begin{observ}
For deterministic $\calG$, 
if $u$ exposes an equation $x_i\doteq H(x_1,\dots,x_n)$ for 
$E,F$ then 
$\eqlevel(E(T_1,\dots,T_n),F(T_1,\dots,T_n))-|u|\leq
\eqlevel(T_i,H(T_1,\dots,T_n))$
(for any $T_1,\dots,T_n$). 
Thus if $E(T_1,\dots,T_n)\sim F(T_1,\dots,T_n)$ then
$T_i\sim H\limtreei(T_1,\dots,T_n))$; 
otherwise (at least)
$T_i\sim_k H\limtreei(T_1,\dots,T_n))$ for 
$k=\eqlevel(E(T_1,\dots,T_n),F(T_1,\dots,T_n))-|u|$.
\end{observ}

\section{A ``word-labelling predicate'' $\models$ and its soundness}\label{subsec:dersystem}

\begin{defn}
(Given $\calG$,)
a \emph{pair}
$(E(x_1,\dots,x_n),F(x_1,\dots,x_n))$ is \emph{sound} if we
have  $E\sigma\sim F\sigma$ 
for all ground instances 
$(E\sigma,F\sigma)$ (i.e., $(E(T_1,\dots,T_n),F(T_1,\dots,T_n))$) 
of the pair $(E,F)$.
A \emph{set} $\calB$ \emph{of pairs} $(E,F)$ is \emph{sound} if each
element is sound. 
\\
By $\instances(\calB)$ we mean the set of all
ground instances of pairs in $\calB$.
\end{defn}
Let us now assume a given \emph{deterministic}
first-order grammar 
 $\calG=(\calN,\act,\calR)$ in normal form, and
a \emph{finite} set $\calB$, called
a \emph{basis}, 
of pairs (of regular terms) $(E,F)$, supposedly sound;
we 
always 
assume $\calB$ contains the pair $(x_1,x_1)$
which is obviously sound.

The derivation (or deduction) system in
Figure~\ref{fig:dersystem1}, assuming $\calG$ and $\calB$,
provides
an inductive definition of the predicate 
\begin{center}
$\modr{\models_{(T_0,U_0)}}\,\subseteq\, \act^*\times
((\gtrees_{\calN}\times \gtrees_{\calN})\cup\{\success,\fail\})$.
\end{center}
So it is in fact a family of predicates, parametrized with the initial
pair $(T_0,U_0)$ of regular \emph{ground} terms. We write just
\modr{$\models$} instead of 
$\models_{(T_0,U_0)}$ when $(T_0,U_0)$ is  clear from context.

We can see that Axiom,
Basic transition rule (1.) and Rejection rule (6.) guarantee that 
if $T_0\not\sim U_0$ and $ua$ is an offending word for 
$(T_0,U_0)$, which implies $(T_0,U_0)\gt{u}(T,U)$ where $T\not\sim_1 U$,
then  $u\models (T,U)$, which can be read as ``$u$ can be labelled
with the pair $(T,U)$'', and thus $\varepsilon \models \fail$; as
expected, 
 $\varepsilon \models \fail$ is intended to mean $T_0\not\sim U_0$.
 
In the case  $T_0\sim U_0$  we are guaranteed that 
if $(T_0,U_0)\gt{u}(T,U)$ then $u\models (T,U)$ but it is not clear
how to use this to conclude that $T_0\sim U_0$.
To this aim we introduce another ``label'': 
$u\models \success$ 
(read ``$u$ can be (also) labelled with $\success$'')
is intended
to imply
that 
\fial{$u$ is \cerv{N}ot an \cerv{O}ffending \cerv{P}refix 
for $(T_0,U_0)$
if $\calB$ is sound}. Thus $\varepsilon\models\success$ 
is intended to imply
\,$T_0\sim U_0$\, if $\calB$ is sound.
Having this in mind, Basis rule (4.) is clear;
 we will later
realize the reason for the condition $u\neq \varepsilon$.
Bottom-up progression rule (5.) is clear as well;
we note in particular that it enables 
to derive
$u\models \success$ when $u\models (T,U)$ and
$T,U$ do not enable any action. 

The most interesting is the 
Limit subterm replacement rule (2.),
with its particular case of Subterm replacement rule.
It allows to label $u$ also with other pairs than those derived just
by Basic transition; so one $u$ can get many pairs of terms as
``labels''.
This is meant to help
to create instances of the
basis and label the respective words with $\success$.
The condition $|v|<|u|$ is important for soundness of 
the predicate $\models$ (wrt its intended meaning).
The symmetry rule 
(3.) could be dropped if we included 
all symmetric cases in the Limit subterm replacement rule.

As usual, we write 
$u\models (T,U)$, or $u\not\models (T,U)$,
if the predicate is true (i.e., derivable) for the triple
$u,T,U$, or not true (not derivable), 
respectively; similarly for $\success$ and $\fail$.

\begin{figure}[t]
\begin{itemize}
\item
(Axiom) $\varepsilon \models (T_0,U_0)$
\item
(Primary derivation (deduction) rules, determining when $u\models (T,U)$)
\begin{enumerate}
\item (Basic transition) 
\\
If $u\models (T,U)$,  $T\sim_1 U$, $(T,U)\gt{a}(T',U')$
then $ua  \models (T',U')$.
\item (Limit subterm replacement) 
\\
If  $u\models (E(T),U)$, $|v|<|u|$,
  $v\models (T,F(T))$, $F\neq x_1$ 
then $u\models (E(F\limtree),U)$.
\\
(A particular case is Subterm replacement:
if $v\models (T,T')$ then $u\models (E(T'),U)$.)

\item (Symmetry)
If  $u\models (T,U)$ then  $u\models
(U,T)$.
\end{enumerate}
\item
(Secondary derivation rules, determining when $u\models\success$ and/or
$u\models\fail$ )
\begin{enumerate}
\setcounter{enumi}{3}
\item (Basis)
If $u\neq\varepsilon$, 
$u\models (T,U)$ and $(T,U)\in\instances(\calB)$ then 
$u\models\success$.
\item
(Bottom-up progression)
If $u\models (T,U)$, $T\sim_1 U$,
and $ua\models \success$
for all $a$ enabled by $T$ (and $U$)
then
$u\models\success$.
\item (Rejection)
If   $u\models (T,U)$ where $T\not\sim_1 U$ then 
$\varepsilon\models \fail$.
\end{enumerate}
\end{itemize}
\caption{Inductive definition of $\models_{(T_0,U_0)}$
(for given $\calG,\base$)
}\label{fig:dersystem1}
\end{figure}

We now recall Observation~\ref{prop:basicquottraces}
and Proposition~\ref{prop:basiceqleveldrop}, 
and show the following generalization  
in our case of det-first-order grammars.
This is the crucial point for showing
soundness (Proposition~\ref{prop:infsoundness}).
(In fact, Point 3. is used later 
 in the 
 completeness proof.)

\begin{prop}\label{prop:indsound}
Given (a det-first-order grammar $\calG$ and) an initial pair $(T_0,U_0)$:
\\
1. If $u\models (T,U)$  then 
$\eqlevel(T_0,U_0)-|u|\leq \eqlevel(T,U)$.
($T_0\sim U_0$ implies $T\sim U$.)
\\
2. If  $T_0\not\sim U_0$ and $u\models (T,U)$ where 
$u\in\pref(\offpl(T_0,U_0))$ 
then
\\
\hspace*{2em}$\eqlevel(T_0,U_0)-|u|=\eqlevel(T,U)$
and
$\offpl(T,U)=u\lquot \offpl(T_0,U_0)$.
\\
3.
If $T_0\sim U_0$ and $u\models (T,U)$ then
$\traces(T)=\traces(U)=u\lquot \traces(T_0) =u\lquot \traces(U_0)$\,.
\end{prop}
\begin{proof}
We proceed by induction on the length of derivations.
The axiom $\varepsilon\models (T_0,U_0)$ trivially satisfies the
conditions. Suppose that the conditions are satisfied for 
all $u\models (T,U)$ derived by derivations upto length $m$, and
consider a derivation deriving  $u'\models (T',U')$ by $m+1$ applications
of the derivation rules. 

If the last rule was 1. (Basic transition),
using $u\models (T,U)$ and $(T,U)\gt{a}(T',U')$, 
then $ua\models (T',U')$ satisfies the conditions by (the induction
hypothesis and) 
Propositions~\ref{prop:basicquottraces},~\ref{prop:basiceqleveldrop}:
We have $\eqlevel(T,U)-1\leq \eqlevel(T',U')$. If  
$ua\in \pref(\offpl(T_0,U_0))$ then $u\in\pref(\offpl(T_0,U_0))$,
hence $\offpl(T,U)=u\lquot \offpl(T_0,U_0)$ and thus
$a\in\pref(\offpl(T,U))$;
this implies
$\offpl(T',U')=a\lquot \offpl(T,U)$
and thus $\offpl(T',U')=(ua)\lquot \offpl(T_0,U_0)$.

If the last rule was Limit subterm replacement (2.),
so from $u\models (E(T),U)$, $|v|<|u|$,
  $v\models (T,F(T))$, $F\neq x_1$ we have derived
$u\models (E(F\limtree),U)$,
then
the conditions 1. and 3. (for  $u\models (E(F\limtree),U)$) 
follow from (the induction hypothesis and)
Propositions~\ref{prop:congrlim} 
and~\ref{prop:congrproperty};  the condition 2. follows
from Point 3. of Observation~\ref{prop:basicoffwords}.

Rule 3. (Symmetry) obviously preserves the conditions. 
\qed
\end{proof}

\begin{cor}
$T_0\not\sim U_0$ iff $\varepsilon \models  \fail$.
\end{cor}

\begin{prop}\label{prop:infsoundness}(Soundness)
\\
If $\varepsilon\models_{(T,U)}\success$ for all
$(T,U)\in\{(T_0,U_0)\}\cup \instances(\calB)$
then $\base$ is sound and $T_0\sim U_0$.
\end{prop}

\begin{proof}
By contradiction.
\\
Suppose the assumption holds but there is some
$\fial{(T'_0,U'_0)}\in\{(T_0,U_0)\}\cup \instances(\calB)$
with the {\fial{least finite eq-level}}.
Take the {\fial{longest (offending) prefix}} $\fial{u}$ 
of some $w\in\offpl(T'_0,U'_0)$ such that 
$u\models_{(T'_0,U'_0)}\success$. 
The rule deriving $u\models_{(T'_0,U'_0)}\success$
could not be the Basis rule since
this supposes $u\neq\varepsilon$ and $u\models_{(T'_0,U'_0)}(T,U)$
where $(T,U)\in\instances{\calB}$ but 
 Point 2. of Proposition~\ref{prop:indsound} implies that
$\eqlevel(T,U)$ is smaller than the eq-level of any pair in
$\instances(\calB)$.
The same Point 2. also implies that 
the deriving rule could not be Bottom-up
progression since this presupposes $ua\models_{(T'_0,U'_0)}\success$
for 
a longer prefix $ua$ of the above $w\in\offpl(T'_0,U'_0)$.
\qed
\end{proof}

\noindent
So we have a sound system, on condition $\calB$ is sound.
But we note that an algorithm surely cannot process infinitely many
pairs  $(T,U)\in \instances(\calB)$ (to show 
$\varepsilon\models_{(T,U)}\success$ for each of them).
Fortunately, it suffices to consider a ``critical instance'' for each pair
$(E,F)$ in
$\calB$ which has the least eq-level among
the ground instances of $(E,F)$:

\begin{defn}
The \emph{critical instance} of a pair 
 $(E(x_1,\dots,x_n),F(x_1,\dots,x_n))$ is the pair 
  $(E\sigma,F\sigma)$ where $\sigma=[L_1/x_1,\dots,L_n/x_n]$
for fresh 
nullary nonterminals $L_i$ (extending $\calG$) such that 
$L_i\not\sim_1 V$ for any $V\neq L_i$; e.g., $L_i$ gets its own
special action $\ell_i$ and a rule $L_i\gt{\ell_i}L_i$.
By $\worstinstances(\calB)$ we mean the set of the critical instances of
the pairs in $\calB$. 
\end{defn}

\begin{prop}\label{prop:worstinstance}
For any $E(x_1,\dots,x_n)$, $F(x_1,\dots,x_n)$, and any $T_1,\dots,T_n$,
\\
$\eqlevel(E(L_1,\dots,L_n),F(L_1,\dots,L_n))\leq 
\eqlevel(E(T_1,\dots,T_n),F(T_1,\dots,T_n))$
\\
if $L_1,\dots, L_n$
do not occur in $E,F$.
\end{prop}

\begin{proof}
Proposition~\ref{prop:exposing} shows that if we had
$k_1=\eqlevel(E(L_1,\dots,L_n),F(L_1,\dots,L_n))>
k_2= \eqlevel(E(T_1,\dots,T_n),F(T_1,\dots,T_n))$ then there were an offending prefix
$u$ for  $(E(T_1,\dots,T_n),F(T_1,\dots,T_n))$ exposing an equation
$x_i\doteq H(x_1,\dots,x_n)$ for
$E,F$, which also means $k_2\geq |u|$; but then $(E(L_1,\dots,L_n),F(L_1,\dots,L_n))
\gt{u}(L_i, H(L_1,\dots,L_n))$ where $L_i\not\sim_1 H(L_1,\dots,L_n)$,
and thus
$k_1\leq |u|$
(a contradiction).
\qed
\end{proof}

\noindent
The next lemma summarizes some important ingredients for the decidability
of trace equivalence for det-first-order grammars, showing that it
is now sufficient to
prove
the completeness of $\models$.

\begin{lemma}
For a det-first-order grammar $\calG$ and a finite set
$\calB$ of pairs of
terms, if
\begin{equation}\label{eq:sound}
\varepsilon\models_{(T,U)}\success \mbox{ for all }
(T,U)\in\{(T_0,U_0)\}\cup \worstinstances(\calB)
\end{equation}
then $\base$ is sound and $T_0\sim U_0$.
Moreover, the condition~\mbox{\textnormal{(\ref{eq:sound})}} 
is semidecidable.
\end{lemma}

\section{Completeness of the predicate $\models$}
\label{subsec:completeness}

\begin{defn}
Given a det-first-order grammar $\calG$, a finite set $\calB$ of pairs 
of regular terms is
a \emph{sufficient basis} if $\calB$ is sound and 
we have:
if $T_0,U_0$ are regular ground terms where $T_0\sim U_0$
then $\varepsilon\models_{(T_0,U_0)} \success$ (wrt $\calG,\calB$).
\end{defn}
We now aim to prove that any det-first-order grammar (in normal form)
has a sufficient basis.
We use implicitly the fact that if $T_0\sim U_0$
then $u\models_{(T_0,U_0)}(T,U)$ implies 
$T\sim U$, 
$\traces(T)=\traces(U)=u\lquot \traces(T_0) =u\lquot \traces(U_0)$,
and 
$\omtraces(T)=\omtraces(U)=u\lquot \omtraces(T_0) =u\lquot \omtraces(U_0)$
(recall Propositions~\ref{prop:indsound} and~\ref{prop:basicquottraces}).
We start with a simple observation:

\begin{prop}\label{prop:infininter}
Given a det-first-order grammar $\calG$ 
and a sound basis $\calB$, 
if for every triple $T_0,U_0,\alpha$ where
$T_0\sim U_0$ and $\alpha\in\omtraces(T_0)$ ($=\omtraces(U_0)$) there is
a prefix $u$ of $\alpha$ such that $u\models_{(T_0,U_0)}\success$ then
$\calB$ is a sufficient basis.
\end{prop}

\begin{proof}
Suppose $T_0\sim U_0$.
For every maximal $w\in\traces(T_0)$ 
(for which there is no $wa\in\traces(T_0)$)
we have $w\models_{(T_0,U_0)}\success$ by using 
Basic transition 
and Bottom-up progression. 
Thus the assumption that each 
$\alpha\in\omtraces(T_0)$ has a prefix $u$ 
for which $u\models_{(T_0,U_0)}\success$ implies that 
$\varepsilon\models_{(T_0,U_0)}\success$, by repeated use of Bottom-up
progression.
\qed
\end{proof}

\noindent
We now show 
a  sufficient condition for the existence
of a sufficient basis (Proposition~\ref{prop:secondsuf}), 
first introducing  some auxiliary notions to this aim.

\begin{defn}
We assume a det-first-order grammar $\calG$. 
Given $(T_0, U_0)$ where  
$T_0\sim U_0$, an infinite trace  $\alpha\in\omtraces(T_0)$ 
is \emph{$s$-bounded} (for $s\in\Nat$)
if it has a nonempty
prefix $u$ such that
$u\models_{(T_0,U_0)}(E(T_1,\dots,T_n), F(T_1,\dots,T_n))$ 
where 
the pair $(E(x_1,\dots,x_n), F(x_1,\dots,x_n))$ is sound and 
$\pressize(E,F)\leq s$.  
\\
\emph{Grammar} $\calG$ is \emph{$s$-bounded} 
if for each
$T_0\sim U_0$ all  $\alpha\in\omtraces(T_0)$ are $s$-bounded.
\\
Grammar $\calG$ \emph{has a stair-base of width}
$n\in\Nat$ 
if 
there is a function $g:\Nat\rightarrow\Nat$ such that: 
for every $T_0\sim U_0$ 
and every $\alpha\in\omtraces(T_0)$
either $\alpha$ is $g(0)$-bounded or there are 
some (unspecified) terms $T_1, \dots, T_{\cerv{n}}$
and infinitely
many nonempty, increasing prefixes $w_0, w_1, w_2,\dots $ of $\alpha$ such
that $w_j\models_{(T_0,U_0)} (\fial{E_j}(T_1,  \dots, T_n), \fial{F_j}(T_1,  \dots, T_n))$
for some (regular) $\fial{E_j},\fial{F_j}$ with
$\pressize(E_j,F_j)<g(j)$, for all $j=0,1,2,\dots$.
\end{defn}
For illustration and later use, we
first note a particular example of $s$-bounded traces where
$s=\pressize(x_1,x_1)$; then we show
that the stair-base property is indeed sufficient.

\begin{defn}\label{def:repeat}
Given $\calG$ and an initial pair $(T_0,U_0)$,
$w\in\act^*\cup\act^{\omega}$
\emph{has a repeat} if there are two
different prefixes $u_1, u_2$, $|u_1|<|u_2|$, of $w$  and some $T,U$ such that 
$u_1\models (T,U)$, $u_2\models (T,U)$.  
(By Subterm replacement we then derive
$u_2\models (U,U)$, where $(U,U)=(x_1\sigma,x_1\sigma)$ for
$\sigma=[U/x_1]$; so if $(x_1,x_1)\in\calB$ then
$u_2\models\success$.)
\end{defn}

\begin{prop}\label{prop:suffi}
A det-first-order grammar $\calG$ has a stair-base of width $0$ 
iff $\calG$ is $s$-bounded for some $s\in\Nat$.
If $\calG$ is $s$-bounded then it has a sufficient
basis.  
\end{prop}

\begin{proof}
The first part follows directly from the definitions (if the width is $0$ 
then $\calG$ is $g(0)$-bounded).
For the second part it suffices to define $\calB$ as
the set of all sound pairs $(E,F)$ with
$\pressize(E,F)\leq s$ (recalling Proposition~\ref{prop:infininter}).
\qed
\end{proof}

\begin{prop}\label{prop:secondsuf}
If  a det-first-order grammar $\calG$ has a stair-base (of some width)
then it is $s$-bounded (for some $s$)
and thus
has a sufficient basis.
\end{prop}

\begin{proof}
Assume a fixed $\calG$ which has a stair-base of width
$n>0$, for a fixed function $g$.
By Proposition~\ref{prop:suffi}, we are done
once we show
that $\calG$ has also  a stair-base of width $n{-}1$.

So let us consider an arbitrary pair $T_0\sim U_0$
and some $\alpha\in\omtraces(T_0)$ which is not $g(0)$-bounded (for
$T_0,U_0$);
there are prefixes $w_0, w_1, w_2,\dots $ of $\alpha$ such
that 
$w_j\models_{(T_0,U_0)} (\fial{E_j}(T_1,  \dots, T_n), \fial{F_j}(T_1,  \dots, T_n))$
where $\pressize(E_j,F_j)<g(j)$. 
$(E_0,F_0)$ is not sound (since $\alpha$ is not $g(0)$-bounded)
but  $E_0(T_1,  \dots, T_n)\sim F_0(T_1,  \dots, T_n)$
(Proposition~\ref{prop:indsound}, Point 1.).
There is thus a
shortest $v$ exposing an equation for $(E_0,F_0)$ 
(recall Proposition~\ref{prop:exposing}); w.l.o.g. we can assume that
the equation is $x_n\doteq H(x_1,\dots,x_n)$ (where $H\neq x_n$), and thus
$w_0v\models_{(T_0,U_0)} (T_n, H(T_1,  \dots, T_n))$.
Since there are only finitely many pairs $(E,F)$ with
$\pressize(E,F)<g(0)$, there is some $s\in\Nat$ determined by 
$\calG$ and $g(0)$ (independent of $T_0,U_0,\alpha$) 
such that $|v|<s$ and $\pressize(H)<s$.

By Limit subterm replacement we get
$w'_j\models (\modr{E'_j}(T_1,  \dots, T_{\cerv{n-1}}), 
\modr{F'_j}(T_1,  \dots, T_{\cerv{n-1}}))$
for $j=0,1,2,\dots$, where
$w'_j=w_{s+j}$, 
$E'_{j}=E_{s+j}[H\limtreen/x_n]$
and
$F'_{j}=F_{s+j}[H\limtreen/x_n]$.
So defining $g'(j)=g(s+j)+f(s)$, for some trivial function $f$
(e.g. $f(s)=2s$, depending on the definition of the size of graph presentations), 
shows that  $\calG$ has a stair-base of width $n{-}1$
(since our reasoning was independent of $T_0,U_0,\alpha$).
\qed
\end{proof}

\noindent
Now we aim to show that any det-first-order
$\calG=(\calN,\act,\calR)$ (in normal form)
has a stair-base (of some width).
We use further auxiliary notions,
recalling 
$w(X,i)$,
$M_0$, $\boundinc$, and the $d$-prefix form ($d=M_0$ in our case).
 
 \begin{defn}\label{def:exposing}
Given $T$ and $w\in\act^*\cup\act^{\omega}$, $w$ \emph{exposes a
subterm} of $T$ \emph{in depth} $d$ if there is a prefix $u$ of $w$ such that 
$P_d^T\gt{u}x_i$ for some $i$. ($P^T_d$ is the $d$-prefix of $T$.)
\\
For  $T\gt{w}$, $w\in\act^*\cup\act^{\omega}$, 
$T$ \emph{sinks by} $w$ if 
$w$ exposes a subterm of $T$ in depth $1$; hence if  
$T$ \emph{does not sink by} 
$w$ then
$Xx_1\dots x_m\gt{w}$ where
$X$ is the root nonterminal of $T$.
\end{defn}
Now we introduce a key ingredient,
the notions of left- and right-balancing segments, with the right- and
left-balancing pivots; we also use $B$ for ranging over ground terms
(which serve as balancing pivots). We assume that
the \emph{underlying}
$\calG$ \emph{is in normal form} (discussing this issue afterwards).

\begin{defn}\label{def:balanc}
A triple $(T,B,v)$ where $|v|=M_0$ is an \emph{$\ell$-balancing
segment} if $T\sim_{M_0}B$, $T\gt{v}$ and 
$T$ does not sink by a proper prefix of $v$. $B$ is called the
\emph{(balancing) pivot} of this segment, an \emph{$r$-pivot}
in this case. 
The \emph{$\ell$-bal-result}  $\lbalres(T,B,v)$
of this segment 
is defined
as follows: if $T=XT'_1\dots T'_m$,
$Xx_1\dots x_m\gt{v}G(x_1,\dots, x_m)$, and 
$B=P^{B}_{M_0}(W_1,\dots,W_n)$ then 
$\lbalres(T,B,v)$ is the pair $(G(V_1,\dots,V_m),F(W_1,\dots,W_n))$ where
$P^{B}_{M_0}\gt{v}F$  
and
$V_i=F_i(W_1,\dots,W_n)$ where $P^{B}_{M_0}\gt{w(X,i)}F_i$
for $i=1,2,\dots, m=\arity(X)$.
$\lbalres(T,B,v)$
can be also presented as 
$(E(W_1,\dots,W_n),F(W_1,\dots,W_n))$ where
$E=G(F_1,\dots,F_m)$.

An \emph{$r$-balancing segment} $(B,T,v)$, with the \emph{$\ell$-pivot} $B$
and $\rbalres(B,T,v)$, is
defined symmetrically. We say just ``\emph{pivot}'' and
``\emph{bal-result}'' when
the side ($\ell$ or $r$) follows from context.
\end{defn}
Informally, Proposition~\ref{prop:balancelabel} 
captures the following simple idea:
if the ``left-hand side'' (lhs-) term does not sink 
by a segment of length
$M_0{-}1$, so it misses the opportunity to expose
a depth-1 subterm by a shortest word, 
then we can balance, i.e., replace its subterms (in depth $1$
originally) by using the
rhs-term ($r$-pivot), achieving
a pair with bounded finite heads and the
same tails, inherited from the $r$-pivot.
By symmetry the same holds for the case of
a non-sinking rhs-term and an $\ell$-pivot.
In what follows we sometimes 
leave implicit the parts of the claims which follow by
symmetry. 

\emph{Remark.}
The normal form assumption on $\calG$ is technically convenient
(though not really crucial). 
Definition~\ref{def:balanc} makes good sense also for
$w(X,i)=\varepsilon$ (recall Proposition~\ref{prop:computwXi}),
but Proposition~\ref{prop:balancelabel} and some later reasoning
would be slightly more complicated. An alternative to the normal
form assumption would 
be adding
a (harmless)
derivation rule in Figure~\ref{fig:dersystem1} 
enabling to replace
an unexposable subterm arbitrarily.

\begin{prop}\label{prop:balancelabel}
Given 
an initial pair $T_0,U_0$,
if $(T,B,v)$ is an $\ell$-balancing segment and
$(V,W)=\lbalres(T,B,v)$ then 
$u\models(T,B)$ implies
$uv\models (V,W)$. 
\end{prop}

\begin{proof}
Suppose the notation from Definition~\ref{def:balanc}. If  
$u\models(T,B)$ then by Basic transition rule we get
$uv\models(G(T'_1,\dots,T'_m),F(W_1,\dots,W_n))$
and  
$u(w(X,i))\models (T'_i,V_i)$ where
$V_i=F_i(W_1,\dots,W_n)$, for $i=1,2,\dots,m$.
Since $|w(X,i)|<M_0=|v|$, by repeated Subterm replacement 
we get
$uv\models(G(V_1,\dots,V_m),F(W_1,\dots,W_n))$.
\qed
\end{proof}

\noindent
Recalling Observations~\ref{prop:prefixchange}
and~\ref{prop:boundtails}, 
we easily observe the following facts.

\begin{observ}\label{prop:balbounds}
\hfill
\\
(1.) The bal-result of an $\ell$-balancing segment $(T,B,v)$ is determined
by the pivot $B$, the word $v$ (of length $M_0$) and by
the root nonterminal of $T$.
\\
(2.) The depth-size of (finite) terms
$F,G,F_i$ in the bal-result as in 
Definition~\ref{def:balanc} 
is bounded by $M_0+\boundinc(M_0)$. 
\\
(3.) The number $n$ of tails in the bal-result 
$(E(W_1,\dots,W_n),F(W_1,\dots,W_n))$
has an upper bound determined by $\calG$ (since $M_0$ is determined by
$\calG$).
\\
(4.)
If $\lbalres(T,B,v)=(V,W)=(E(W_1,\dots,W_n),F(W_1,\dots,W_n))$
where $E=G(F_1,\dots,F_m)$ as in 
Definition~\ref{def:balanc},  
and $V\gt{w}$ where $w$ exposes a subterm of $V$ in depth 
$(1+\boundinc(M_0))$ then $w$ necessarily exposes a
subterm of $V$ (of the form $F_i(W_1,\dots,W_n)$) 
which is reachable from $B$ by some $w(X,i)$ (of length
$<M_0$).
\end{observ}
We now try to use the possibility of balancing along
an infinite $\alpha\in\omtraces(T_0)$, for a pair $T_0\sim U_0$,
to show that 
$\alpha$ allows
a stair-base of width $n$, with a function $g$, 
which are independent of $T_0,U_0,\alpha$
(i.e., with $n$ and $g$ determined just by grammar $\calG$).
We first observe that if there are only
finitely many balancing opportunities then
$\alpha$ allows a repeat (recall Definition~\ref{def:repeat})
and the condition is trivial:

\begin{defn}\label{def:nextbal}
Assume an initial pair $T_0\sim U_0$ and a fixed 
$\alpha\in\omtraces(T_0)$. 
\\
For (the triple $u,T,U$ such that)
$u\models_{(T_0,U_0)} (T,U)$, $u$ being a prefix of
$\alpha$, we define the 
\emph{next $\ell$-segment} as the $\ell$-balancing segment $(T',B,v)$ 
for the shortest $w$ (if there is some)
such that $uwv$ is a prefix of $\alpha$ and $(T,U)\gt{w}(T',B)$. 
The \emph{distance of} this next $\ell$-segment is defined as $|w|$. 
Similarly we define the 
\emph{next $r$-balancing segment} for  $u\models_{(T_0,U_0)} (T,U)$.
\end{defn}

\begin{prop}
Given  $T_0,U_0,\alpha$ as in Definition~\ref{def:nextbal},
if there is no next $\ell$-segment and no next $r$-segment for some
$u\models_{(T_0,U_0)} (T,U)$, $u$ being a prefix of $\alpha$,
then $\alpha$ has a repeat.
\end{prop}

\begin{proof}
Consider performing $\alpha'=u\lquot \alpha=a_1a_2a_3\dots$ from $T$;
let $T=T_1\gt{a_1}T_2\gt{a_2}T_3\gt{a_3}\cdots$.
If there were a (first) segment $T_i\gt{a_{i}}T_{i+1}\gt{a_{i+1}}\dots
\gt{a_{i+M_0-1}}T_{i+M_0}$ where $T_i$ is a subterm of $T$ but none of 
$T_{i+1}, T_{i+2}, \dots, T_{i+M_0}$ is a subterm of $T$ (of $T_i$, in
fact) then we had an $\ell$-balancing segment. Hence each $T_i$ is
reachable from a subterm of $T$ by a word of length $<M_0$, which
means that there are only finitely many different $T_i$. Similarly for
$U_i$ on the rhs. This guarantees a repeat.
\qed
\end{proof}

\begin{prop}~\label{prop:ghasstairbase}\hfill
\\
Any det-first-order grammar $\calG$ in normal form
has
a stair-base (of some width).
\end{prop}

\begin{proof}
We assume a det-first-order grammar $\calG=(\calN,\act,\calR)$ in
normal form,
a pair $T_0\sim U_0$ and a fixed $\alpha\in\omtraces(T_0)$; 
we further write
$\models$ instead of $\models_{(T_0,U_0)}$. 
Assuming that $\alpha$ has no repeat, we show that it has
a stair-base of
width $n$, with function $g$, where $n,g$ are independent of
$T_0,U_0,\alpha$.

We will present $\alpha$ as $u_1v_1u_2v_2u_3v_3\dots$ where $|v_i|=M_0$,
attaching to each $v_i$ a triple $(side_i,T_i,U_i)$ and a pair 
$(T'_i,U'_i)$
such that
$side_i\in\{\ell,r\}$, $(T_i,U_i,v_i)$ is a $side_i$-balancing
segment,  $(T'_i,U'_i)=\sideibalres(T_i,U_i,v_i)$,
$(T_0,U_0)\gt{u_1}(T_1,U_1)$ and
$(T'_i,U'_i)\gt{u_{i+1}}(T_{i+1},U_{i+1})$ for $i=1,2,3,\dots$.
Hence  $u_1v_1\dots u_{i-1}v_{i-1}u_i\models (T_i,U_i)$ and
 $u_1v_1\dots u_{i}v_{i}\models (T'_i,U'_i)$.
We note that each $v_i$ has the corresponding pivot
$B_i$, i.e. one of $T_i, U_i$, depending on $side_i$.

$(T_1,U_1,v_1)$ is defined as the next $\ell$-segment
or the next $r$-segment for $\varepsilon\models (T_0,U_0)$; if both
exist, the one with the 
smaller distance is chosen
(recall Definition~\ref{def:nextbal}), and  we prefer
$\ell$, say, to break ties. This also induces $u_1$ 
(where $(T_0,U_0)\gt{u_1}(T_1,U_1)$). 

Suppose $u_1v_1\dots u_iv_i$ 
have been
defined, and assume that $(T_i,U_i,v_i)$ is an $\ell$-segment 
(so $B_i=U_i$ is an $r$-pivot; the other case is symmetrical). 
If for $u_1v_1\dots u_{i}v_{i}\models (T'_i,U'_i)$
there is the next $\ell$-segment
with the distance at most
\begin{equation}\label{eq:Mone}
M_1=(1+\boundinc(M_0))\cdot M_0
\end{equation}
then we use 
this segment to define $u_{i+1}v_{i+1}$ etc.; there was no switch, we
have $side_{i+1}=side_i=\ell$.
If there is no such ``close'' 
$\ell$-segment (since the $\ell$-side terms keep sinking),
we note
that a subterm of $T'_i$ in depth $(1+ \boundinc(M_0))$ has been
exposed by $w$ where $|w|=M_1$ and 
$u_1v_1\dots u_{i}v_{i}w$ is a prefix of $\alpha$;
let $(T_i',U_i')\gt{w}(T_i'',U_i'')$. 
Point 4. in Observation~\ref{prop:balbounds} implies that 
$T_i''$ is reachable from $B_i$, by a word arising from $v_iw$ by
replacing a prefix $v_iw'$ with some $w(X,j)$.
Here we define 
$(T_{i+1},U_{i+1},v_{i+1})$ as the 
next $\ell$- or $r$-segment
with the least distance for 
$u_1v_1\dots u_{i}v_{i}w\models (T_i'',U_i'')$ (so $w$ is a
prefix of $u_{i+1}$).
This might, but also might not, mean a switch of the pivot side.

Anyway,
$B_{i+1}$ is reachable from $B_i$ by $w_i$ where either $w_i=v_iu_{i+1}$ 
or $w_i$ arises from $v_iu_{i+1}$ by replacing a prefix
$v_iw'$
of length $\leq M_0+M_1$
by a (shorter) word $w(X,j)$.
We thus get a \emph{pivot-path}
$$B_1\gt{w_1} B_2\gt{w_2}B_3\gt{w_3}\cdots$$
We note that if some $w_i$ is longer than $M_0+M_1$, $w_i=w'_iw''_i$
where $|w'_i|=M_0+M_1$, then the ``pivot-path sinks''
in any segment of $w''_i$ of length $M_0$, i.e.:
for any partition $w''_i=w''_{i1}vw''_{i2}$, $|v|=M_0$, we have 
$B_i\gt{w'_iw''_{i1}}W_1\gt{v}W_2\gt{w''_{i2}}B_{i+1}$ where 
$W_1$ sinks by $v$.

This implies for any 
$B_1\gt{u}V\gt{\beta'}$ where $u\beta'=\beta=w_1w_2w_3\dots$
that there is a nonempty prefix $u'$ of $\beta'$ of length at most
\begin{equation}\label{eq:Mtwo}
M_2= (M_0+M_1) + (1+\boundinc(M_0+M_1))\cdot M_0
\end{equation}
such that
$V\gt{u'}B_i$ (for some $i$) or
$V$ sinks by $u'$. (Informally: any segment $V\gt{w}$
of the pivot path $B_1\gt{\beta}$ with length $|w|=M_2$
either contains a pivot $B_i$ or sinks.)
Hence if $\beta$ exposes subterms of $B_1$ in all depths
then infinitely many pivots $B_i$ are
equal (since reachable from subterms of $B_1$ by words of length
$\leq M_2$);
Point 1. in Observation~\ref{prop:balbounds}  shows 
that $\alpha$ would then have a repeat.

So there is the maximal depth $d$ such that $\beta=w_1w_2w_3\dots$ exposes
a (unique) subterm $V_1$ of $B_1$ in depth $d$; hence
$B_1\gt{u}V_1\gt{\beta'}$ where $u\beta'=\beta$ and 
$V_1$ does not sink by $\beta'$.
Let  $V_1=P_{M_0}^{V_1}(T_1,\dots,T_n)$ be the $M_0$-prefix form of
$V_1$, and let $k\in\Nat$ be the least such that 
$B_1\gt{u}V_1\gt{u'}B_k\gt{w_k}B_{k+1}\gt{w_{k+1}}\cdots$
($u'$ being a prefix of $\beta'$).

Then pivots $B_{k+j}$, $j=0,1,2,\dots$, are of the form 
 $G_{j}(T_1,\dots,T_n)$ where $G_j$ are finite terms in which
 each occurrence 
 of a variable has depth $M_0$
 at least. Moreover, $\depthsize(G_j)\leq M_0+\boundinc(M_2)\cdot
 (j+1)$ (by the above ``contains a pivot or sinks'' fact).
 
 Hence
the bal-results $(T'_{k+j},U'_{k+j})$ for $B_{k+j}$, $j=0,1,2,\dots$
are of the form
$(E_j(T_1,\dots,T_n),F_j(T_1,\dots,T_n))$ 
where $E_j,F_j$ are finite terms with the depth-size bounded by
$g'(j)$ for some $g'$ determined by
the grammar $\calG$ (recall Definition~\ref{def:balanc}, Point 2. in
Observation~\ref{prop:balbounds}, and the fact that 
$M_0, M_1, M_2, \boundinc$ are determined by
$\calG$).
There is thus $g:\Nat\rightarrow\Nat$
(independent of $T_0,U_0,\alpha$) such that 
$\pressize(E_j,F_j)<g(j)$, for $j=0,1,2,\dots$. 

Point 3. in
Observation~\ref{prop:balbounds} thus
implies that $\calG$ has a
stair-base (of some width).
\qed
\end{proof}

\noindent
In fact, we have thus shown the next completeness
lemma, and the main theorem.

\begin{lemma}\label{lem:completeness1} (Completeness)
\\
For each det-first-order grammar $\calG$ in normal form
there is a sufficient (sound) basis $\base$.
\end{lemma}

\begin{theorem}\label{th:tracedecid}
Trace
equivalence for deterministic
first-order grammars 
is decidable. 
\end{theorem}
For deciding $T_0\stackrel{?}{\sim} U_0$\,, 
an algorithm based on soundness and completeness is clear
(using the effective manipulations with graph 
presentations of regular terms): when we
are allowed to generate any finite basis for a given initial pair
$T_0,U_0$ then both questions ``$T_0\not\sim U_0$~?'', 
``$T_0\sim U_0$~?'' are semidecidable; when verifying $T_0\sim U_0$,
we have to verify all (critical instances of) pairs included in the basis as well.

\emph{Remark.} By inspecting the proofs we could note that a sufficient
basis for a det-first-order grammar (in normal form) is, in fact,
computable (since we now know that
the value $s$ determined by $\calG$ and $g(0)$ in
the proof of Proposition~\ref{prop:secondsuf} is computable) 
but this
computability does not seem much helpful.

\subsection*{Conclusions}

The presented proof of the decidability of trace equivalence for
det-first-order grammars routinely applies 
to the dpda language equivalence, as also shown in Appendix 1. 
The novelty here is the presentation in the framework of first order terms,
resulting in a proof which seems technically simpler than the previous
ones.

Appendix 2. gives another look at the complexity result by 
Stirling~\cite{Stir:DPDA:prim}, showing that the
framework of first-order terms can be useful there as well.

\bibliographystyle{plain}
\bibliography{root}

\noindent
\textbf{Appendix 1.}

\section{DPDA 
language
equivalence problem presented via
trace equivalence for det-first-order grammars}

A \emph{deterministic pushdown automaton} (\emph{dpda})
is a tuple 
$M=(Q,\act,\Gamma,\Delta)$ consisting of
finite sets $Q$ of (control) states, 
$\act$ of actions (or terminals),
$\Gamma$ of stack symbols, and $\Delta$ of transition rules.
For each pair $pA$, $p\in Q$, $A \in \Gamma$, and each 
$a\in\act\cup\{\varepsilon\}$, $\Delta$ contains at most one rule
of the type  $pA \gt{a}q\alpha$, where $q\in Q$, $\alpha\in\Gamma^*$.
Moreover, any pair $pA$ is (exclusively) either \emph{stable}, 
i.e. 
having no rule 
$pA\gt{\varepsilon}q\alpha$, or
\emph{unstable}, in which case there
is (one rule $pA\gt{\varepsilon}q\alpha$ and) no rule
$pA\gt{a}q\alpha$ with $a\in\act$.

A dpda $M$ generates a labelled transition system
$(Q\times \Gamma^*,\act\cup\{\varepsilon\}, 
\{\gt{a}\}_{a\in\act\cup\{\varepsilon\}})$
where the states are
configurations  $q\alpha$ ($q\in Q$, $\alpha\in\Gamma^*$).
Having our grammars in mind, we view a rule $pA\gt{a}q\alpha$
as $pAx\gt{a}q\alpha x$ (for a formal variable $x$), inducing
$pA\beta \gt{a}q\alpha\beta$ for every $\beta\in\Gamma^*$. 
The transition relation is extended to words $w\in\act^*$ 
as usual; we note that $p \alpha \gt{w} q \beta$ 
can comprise more than  $|w|$ basic steps, due
to possible ``silent'' $\varepsilon$-moves.
Each configuration $p \alpha$
has its associated \emph{language}
$L(p \alpha)=\{w\in\act^*\mid p \alpha \gt{w}q \varepsilon$
for some $q \}$.
The \emph{dpda language equivalence problem} is:  
given a dpda $M$ and two configurations
$p \alpha$, $q \beta$,
is $L(p \alpha)=L(q\beta)$~?

\emph{Remark.}
It is straightforward to observe that this setting is
equivalent to the classical problem of language equivalence 
between deterministic pushdown automata with 
accepting states.
First, the disjoint
union of two dpda's is a dpda. Second, for languages
$L_1,L_2\subseteq\Sigma^*$ we have 
$L_1=L_2$ iff
$L_1\cdot\{\$\}=L_2\cdot\{\$\}$, for an endmarker $\$\not\in\Sigma$;
so restricting to prefix-free deterministic context-free languages,
accepted by dpda via empty stack,
does not mean losing generality.

Each dpda $M$ can be transformed by a standard polynomial-time
algorithm so that
all $\varepsilon$-transitions  are popping, i.e., of the type
$pA\gt{\varepsilon}q$, while  
$L(pA\alpha)$, for stable $pA$, keep unchanged.
(A principal point is that 
a rule $pA\gt{\varepsilon}qB\alpha$ where $qB\gt{a_1}q_1\beta_1$,
$\dots$,  $qB\gt{a_k}q_k\beta_k$ can be
replaced with rules $pA\gt{a_j}q_j\beta_j\alpha$; unstable pairs
$pA$ enabling only an infinite sequence of $\varepsilon$-steps are
determined and removed.) 

It is also harmless to assume that for each stable $pA$ and each
$a\in\act$ we have one rule $pA\gt{a}q\alpha$ 
(since we can introduce a `dead' state $q_d$ with rules $q_dA\gt{a}q_dA$
for all $A\in \Gamma, a\in\act$, and for every 
`missing' rule $pA\gt{a}..$ we add  $pA\gt{a}q_dA$). 
$L(p\alpha)$ are unchanged by this transformation.
Then  $w\in \act^*$ is not enabled in $p\alpha$ iff $w=uv$ where
$p\alpha\gt{u}q\varepsilon$ (for some $q$), so $u\in L(p\alpha)$, and
$v\neq\varepsilon$.
This reduces language equivalence to trace equivalence:
\begin{center}
$L(p\alpha)=L(q\beta)$ \ iff
\  $\forall w\in\act^*:p\alpha\gt{w}\,\Leftrightarrow\,q\beta\gt{w}$.
\end{center}

\begin{prop} \label{dpdaprop}
The dpda language equivalence problem 
is polynomial-time reducible 
to the deterministic first-order grammar 
equivalence problem.
\end{prop}

\begin{proof}
Assume an ($\varepsilon$-popping) dpda $M=(Q,\act,\Gamma,\Delta)$
transformed as above
(so trace equivalence coincides with language equivalence).
We define  the first-order grammar $\calG_M=(\calN,\act,\calR)$ where 
$\calN=\{pA\mid pA $ is \emph{stable}$\}\cup\{\bot\}$; 
each $X=pA$ gets arity $m=|Q|$, and 
$\bot$ is a special nullary nonterminal not
enabling any action.
A dpda configuration $p \alpha$ 
is transformed to the term 
$\calT(p \alpha)$ defined inductively by rules 1.,2.,3. below,
assuming $Q=\{q_1,q_2,\dots,q_m\}$. 
\begin{enumerate}
\item
$\calT(q\varepsilon)=\bot$. 
\item
If $qA\gt{\varepsilon}q_i$ ($qA$ is unstable)
then $\calT(qA\beta)=\calT(q_i\beta)$.
\item
If $qA$ is stable then 
$\calT(qA\beta)=X\,\calT(q_1\beta) \dots \calT(q_m\beta)$
where $X=qA$.
\item
$\calT(q_ix)=x_i$.
\end{enumerate}
Rule 4. is introduced to enable  the smooth
transformation of a dpda rule  $pA\gt{a}q\alpha$, where $a\in\act$, 
rewritten in the form
 $pAx\gt{a}q\alpha x$,  to 
the $\calG_M$-rule
$\calT(pAx)\gt{a} \calT(q\alpha x)$, i.e. to
$Xx_1\dots x_m\gt{a} \calT(q\alpha x)$, where $X=pA$.
Thus $\calR$ in $\calG_M$ is defined (with no $\varepsilon$-moves).
We observe easily: 
if $p A\alpha \gt{\varepsilon} q \alpha$ 
(recall that $\varepsilon$-steps are popping) then
$\calT(pA\alpha)=\calT(q \alpha)$; 
if $p A\alpha \gt{a} q \beta\alpha$ ($a\in\act$, $pA$ stable)
then
$\calT(pA\alpha)\gt{a}\calT(q \beta\alpha)$.
This also implies: if $p\alpha\gt{w}q\varepsilon$ then 
$\calT(p\alpha)\gt{w}\bot$. Thus 
\begin{center}
$L(p\alpha)=L(q\beta)$ 
\ iff
\ $\big(\forall
w\in\act^*:p\alpha\gt{w}\,\Leftrightarrow\,q\beta\gt{w}\big)$
\ iff
\ $\calT(p\alpha)\sim \calT(q\beta)$. 
\end{center}
We note that $\calT(q\alpha)$ can have (at most)
$1+m+m^2+m^3+\cdots +m^{|\alpha|}$ subterm-occurrences,
but the natural finite graph presentation of $\calT(q\alpha)$
has at most 
 $1+m(|\alpha|-1)+1$ nodes and can be obviously constructed in
 polynomial time.
\qed
\end{proof}

\noindent
\textbf{Appendix 2.}

\vspace{-4mm}

\section{ A complexity bound}
\label{sec:complexitybound}

We have, in fact, not fully used the potential of the
pivot-path
$B_1\gt{w_1} B_2\gt{w_2}B_3\gt{w_3}\cdots$ discussed in the proof of 
Proposition~\ref{prop:ghasstairbase}.
For showing the existence of a sufficient basis
it was
sufficient to use just the stair-base subterm $V_1$ of $B_1$. 
We will now explore the idea of the described balancing strategy
further, which allows to derive a concrete computable function bounding the
length of potential offending words (i.e., the eq-level) when given
(a det-first-order grammar $\calG$ and) an initial pair $(T_0,U_0)$.

For our aims 
(in the context of upper complexity bounds),
an \emph{elementary function}
is a function (of the type $\Nat^k\rightarrow \Nat$)
arising by a finite composition of
constant
functions,
the elementary operations $+, -,\times, \div$\,, 
and the exponential operator $\uparrow$, where 
$a\uparrow n=a^n$\,.
When we say that a number is \emph{simply bounded}, we mean that there
is an elementary function of the size of $\calG$ giving an upper
bound; e.g., the constants $M_0,M_1,M_2$, the number of tails in the
$M_0$-prefix form, the depth-size of the heads in any balancing
result, etc., are obviously simply bounded.

The ``first'' nonelementary (hyper)operator is 
\emph{iterated exponentiation} $\uparrow\uparrow$,
also called \emph{tetration}:
$a\uparrow\uparrow n= 
a\uparrow (a\uparrow (a\uparrow ( \dots a \uparrow a )\dots ))$
where $\uparrow$ is used $n$-times.

Our analysis will yield the following bound on the length of
offending words, which has an obvious algorithmic consequence:

\begin{theorem}\label{th:lengthbound}
For any triple $\calG$, $T_0,U_0$ 
with the size $\inputsize$ (of a standard presentation),
where $\calG$ is a det-first-order grammar and 
$(T_0,U_0)$ is an initial pair such that $T_0\not\sim U_0$,
there is a sequence of actions 
which is enabled in just one of
 $T_0, U_0$ and its length is bounded by
$2\uparrow\uparrow f(\inputsize)$, where $f$ is an elementary
function independent of $\calG$, $T_0, U_0$.
\end{theorem}

\begin{cor}
Trace 
equivalence 
for deterministic first-order grammars
can be decided 
in time (and space) 
$O(2\uparrow\uparrow g(\inputsize))$
for an elementary function $g$.
\end{cor}
The analogous claims hold for language
 equivalence of deterministic pushdown automata, as follows from the
 reduction in the previous section.

\noindent
We now aim to prove Theorem~\ref{th:lengthbound}.
In the rest of this section 
we assume a fixed
det-first-order grammar
$\calG=(\calN,\act,\calR)$ in normal form, if not said otherwise.

Later we will consider a fixed initial pair 
$(T_0,U_0)$, where $T_0\not\sim U_0$, and 
a fixed offending word 
$\alpha\in\act^*$ for $(T_0,U_0)$. 
The word $\alpha$ is (now) finite, and 
our aim is to show an appropriate upper bound on its length 
which will prove 
Theorem~\ref{th:lengthbound}.
To this aim, we first show an ``upper-bound tool'' 
(Proposition~\ref{prop:afixedpair} with
Corollary~\ref{cor:repeateqlevel}),
and then we 
demonstrate that the balancing strategy along our finite $\alpha$
(the same strategy as used
in the proof of 
Proposition~\ref{prop:ghasstairbase} along the infinite $\alpha$ there)
guarantees that the upper-bound tool can be applied to bound the
length of $\alpha$.

We start with noting a possible new type of
sound subterm replacement; roughly speaking, 
if a pair of heads repeats then an equation (if some is exposable)
is available 
for a potential application. (This new subterm replacement serves just
for our reasoning, we will not extend the definition of $\models$.) 
In the next proposition, 
it might help to imagine that we have 
$u\models  (E(U_1\dots U_n),F(U_1\dots U_n))$ and 
 $v\models (E(V_1\dots V_n),F(V_1\dots V_n))$ where $|u|<|v|$ and we
would like to label $v$ also with  
$(E(V'_1\dots V'_n),F(V'_1\dots V'_n))$, where $V'_j$ arises from
$V_j$ by possible replacing
of some occurrences of $U_n$ with 
$H\limtreen(U_1,\dots,U_{n-1})$.

\begin{prop}\label{prop:newreplacement}
Assume a pair $(E(x_1,\dots,x_n), F(x_1,\dots,x_n))$ (of regular
terms) for which there is a shortest
$w\in\act^*$ 
exposing an equation, w.l.o.g. say $x_n\symbeq H(x_1,\dots,x_n)$ 
($H\neq x_n$).
Let us consider the following three pairs 
\begin{center}
$(E(U_1\dots U_n),F(U_1\dots U_n))$,
\ $(E(V_1\dots V_n),F(V_1\dots V_n))$,
\ $(E(V'_1\dots V'_n),F(V'_1\dots V'_n))$,
\end{center}
with eq-levels $k_1,k_2,k_3$, respectively, where
$V_j=G_j(U_n)$ and $V'_j=G_j(H\limtreen(U_1,\dots,U_{n-1}))$
(for $j=1,2,\dots,n$).
Then $k_3\geq \min\{k_1,k_2\}$, and $k_3=k_2$ if $k_1>k_2$.
\end{prop}

\begin{proof}
We note that 
$U_n\sim_{k} H(U_1,\dots, U_{n})\sim_{k} 
H\limtreen(U_1,\dots, U_{n-1})$ where $k=max\{k_1{-}|w|,0 \}$;
hence $V_j\sim_{k} V'_j$ for $j=1,2,\dots,n$.

Suppose now $k_3<\min\{k_1,k_2\}$. Then 
there is an offending prefix $w'$ for the third pair 
which exposes an equation 
for $E,F$ (recall Proposition~\ref{prop:exposing});
necessarily $|w'|\geq |w|$.
We thus have
$(E(V'_1,\dots,V'_n),F(V'_1,\dots,V'_n))\gt{w'} 
(V'_i,G(V'_1,\dots,V'_n))$ (or vice versa)
for some $i$ and $G$,
where
$\eqlevel(V'_i,G(V'_1,\dots,V'_n))= k_3-|w'|=k'$. 
But
$\eqlevel(V_i,G(V_1,\dots,V_n))\geq
k_2{-}|w'|\geq k'{+}1$ 
and $V_i\sim_{k'+1}V'_i$  (since $k=k_1{-}|w|>k'$), 
which yields a contradiction.
Similarly we would contradict the case $k_2<\min\{k_1,k_3\}$.
Hence the claim follows.
\qed
\end{proof}

\noindent
A simple corollary is that if $u\models (E(T),F(T))$ and 
 $v\models (E(e(T)),F(e(T)))$, for a regular term
 $e=e(x_1)$ (a ``$1$-tail extension''), where $|u|<|v|$, then
 if $v$ is an offending prefix for
 $(T_0,U_0)$
then
$\eqlevel(E(e(T)),F(e(T)))$ is independent of $T$
(since $\eqlevel(E(e(T)),F(e(T)))=\eqlevel(E(e(H\limtreeone)),F(e(H\limtreeone)))$
for the appropriate $H$).
Hence if we also have $u'\models (E(T'),F(T'))$ and
 $v'\models (E(e(T')),F(e(T')))$, where $|u'|<|v'|<|v|$ 
then it is impossible that both $v,v'$ are offending prefixes
for $(T_0,U_0)$.

We now show a generalization, which seems related to
the Subwords Lemma in~\cite{Senizergues:ICALP03}.
We use a visually more convenient ``two-dimensional'' 
notation for (composed) terms: the first rectangle below is a
shorthand for $E\sigma_1\sigma_2\cdots\sigma_r$ 
where $\sigma_j=[e^{i_j}_1/x_1,\dots, e^{i_j}_n/x_n]$; it also
presupposes that the variables occuring in all terms in the rectangle
are from the set $\{x_1,\dots,x_n\}$.

\noindent
\begin{minipage}{0.6\textwidth}
Given a (head) pair
$(E(x_1,\dots,x_n), F(x_1,\dots,x_n))$, and 
$n$ tuples, called  \emph{head extensions}, 
$(e^1_1,\dots,e^1_n)$, $(e^2_1,\dots,e^2_n)$, $\dots$, 
$(e^n_1,\dots,e^n_n)$, where $E,F$ and 
$e^i_j= e^i_j(x_1,\dots,x_n)$ are regular terms, 
we call $(E',F')$ an \emph{extended head pair} if 
it can be presented as depicted, 
for $0\leq r\leq n$ and $1\leq i_1<i_2<\cdots < i_r\leq n$.
We note that there are $2^n$ such presentations.
By $(E'_{max},F'_{max})$, called the \emph{maximal pair},
we denote
the pair with $r=n$ and $i_1=1, i_2=2, \dots, i_n=n$.
\end{minipage}
\begin{minipage}{0.4\textwidth}
\begin{center}
\framebox{
\blaytree{15}{
$E$

$e^{i_1}_1 \dots e^{i_1}_n$

$e^{i_2}_1 \dots e^{i_2}_n$

$\dots$

$e^{i_r}_1 \dots e^{i_r}_n$
}
}
\framebox{
\blaytree{15}{
$F$

$e^{i_1}_1 \dots e^{i_1}_n$

$e^{i_2}_1 \dots e^{i_2}_n$

$\dots$

$e^{i_r}_1 \dots e^{i_r}_n$
}
}
\end{center}
\end{minipage}

\begin{prop}\label{prop:afixedpair}
Assume  
$(E(x_1,\dots,x_n), F(x_1,\dots,x_n))$, 
$(e^1_1,\dots,e^1_n)$, $(e^2_1,\dots,e^2_n)$, $\dots$, 
$(e^n_1,\dots,e^n_n)$ as above, 
and consider a tuple  $T_1,\dots,T_n$ (of regular ground
terms).
\\
If $k=\eqlevel(E'_{max}(T_1,\dots,T_n),F'_{max}(T_1,\dots,T_n))$ is
less than 
\\
$\eqlevel(E'(T_1,\dots,T_n),F'(T_1,\dots,T_n))$
for any other extended head pair $(E',F')$ then 
$k$ is independent of  $T_1,\dots,T_n$.
\end{prop}

\begin{proof}
The claim is trivial for $n=0$. We prove it for $n > 0$, assuming
it holds for $n{-}1$.
\\
If  there is no $w\in\act^*$ exposing an equation 
for $E,F$ then the claim is trivial since
$\eqlevel(E(U_1,\dots,U_n),F(U_1,\dots,U_n))$
is independent of 
$U_1,\dots,U_n$ (for any $U_1,\dots,U_n$). 
So we assume
a shortest $w\in\act^*$ exposing an equation 
for $E,F$, w.l.o.g. 
$x_n\symbeq H(x_1,\dots,x_n)$.

\medskip

\noindent
\begin{minipage}{0.6\textwidth}
Each pair $E'(T_1,\dots,T_n), F'(T_1,\dots,T_n)$ 
where $E',F'$ is an extended head pair with 
$i_1=1$
gives rise to
the depicted pair, by replacing
each $e^{i_j}_{\ell}(x_1,\dots,x_n)$ 
with 
$\bar{e}^{i_j}_{\ell}(x_1,\dots,x_{n-1})=
{e}^{i_j}_{\ell}[H\limtreen/x_n]$
and by omitting the now superfluous 
$T_n$ and $e^{i_j}_n$ (for $i_j\neq 1$). This procedure is independent
of trees $T_1,\dots,T_n$; these are handled as ``black boxes''.
\end{minipage}
\begin{minipage}{0.4\textwidth}
\begin{center}
\framebox{
\blaytree{15}{
$E$

$\bar{e}^{1}_1 \dots \bar{e}^{1}_n$

$\bar{e}^{i_2}_1 \dots \bar{e}^{i_2}_{n-1}$

$\dots$

$\bar{e}^{i_r}_1 \dots \bar{e}^{i_r}_{n-1}$

$T_1 \dots T_{n-1}$
}
}
\framebox{
\blaytree{15}{
$F$

$\bar{e}^{1}_1 \dots \bar{e}^{1}_n$

$\bar{e}^{i_2}_1 \dots \bar{e}^{i_2}_{n-1}$

$\dots$

$\bar{e}^{i_r}_1 \dots \bar{e}^{i_r}_{n-1}$

$T_1 \dots T_{n-1}$
}
}
\end{center}
\end{minipage}

\medskip

\noindent
We thus get $2^{n-1}$ pairs, with the head pair 
$(\bar{E},\bar{F})=
(E(\bar{e}^{1}_1 \dots \bar{e}^{1}_n),
F(\bar{e}^{1}_1 \dots \bar{e}^{1}_n))$ and head extensions
$(\bar{e}^{2}_1, \dots, \bar{e}^{2}_{n-1})$, 
$(\bar{e}^{3}_1, \dots, \bar{e}^{3}_{n-1})$, 
$\dots$,
$(\bar{e}^{n}_1, \dots, \bar{e}^{n}_{n-1})$.
Repeated use of Proposition~\ref{prop:newreplacement}
for the triples of the form 

\framebox{
\blaytree{15}{
$E$

$\bar{e}^{i_{\ell+1}}_1 .. \bar{e}^{i_{\ell+1}}_n$

$\dots$

$\bar{e}^{i_r}_1 \dots \bar{e}^{i_r}_{n}$

$T_1 \dots T_{n}$
}
}
\framebox{
\blaytree{15}{
$F$

$\bar{e}^{i_{\ell+1}}_1 .. \bar{e}^{i_{\ell+1}}_n$

$\dots$

$\bar{e}^{i_r}_1 \dots \bar{e}^{i_r}_{n}$

$T_1 \dots T_{n}$
}
}
\,,\,
\framebox{
\blaytree{15}{
$E$

${e}^{i_1}_1 \dots {e}^{i_1}_n$

$\dots$

${e}^{i_{\ell-1}}_1 .. {e}^{i_{\ell-1}}_n$

${e}^{i_{\ell}}_1 \dots {e}^{i_{\ell}}_n$

$\bar{e}^{i_{\ell+1}}_1 .. \bar{e}^{i_{\ell+1}}_n$

$\dots$

$\bar{e}^{i_r}_1 \dots \bar{e}^{i_r}_{n}$

$T_1 \dots T_{n}$
}
}
\framebox{
\blaytree{15}{
$F$

${e}^{i_1}_1 \dots {e}^{i_1}_n$

$\dots$

${e}^{i_{\ell-1}}_1 .. {e}^{i_{\ell-1}}_n$

${e}^{i_{\ell}}_1 \dots {e}^{i_{\ell}}_n$

$\bar{e}^{i_{\ell+1}}_1 .. \bar{e}^{i_{\ell+1}}_n$

$\dots$

$\bar{e}^{i_r}_1 \dots \bar{e}^{i_r}_{n}$

$T_1 \dots T_{n}$
}
}
\,,\,
\framebox{
\blaytree{15}{
$E$

${e}^{i_1}_1 \dots {e}^{i_1}_n$

$\dots$

${e}^{i_{\ell-1}}_1 .. {e}^{i_{\ell-1}}_n$

$\bar{e}^{i_{\ell}}_1 \dots \bar{e}^{i_{\ell}}_n$

$\bar{e}^{i_{\ell+1}}_1 .. \bar{e}^{i_{\ell+1}}_n$

$\dots$

$\bar{e}^{i_r}_1 \dots \bar{e}^{i_r}_{n}$

$T_1 \dots T_{n}$
}
}
\framebox{
\blaytree{15}{
$F$

${e}^{i_1}_1 \dots {e}^{i_1}_n$

$\dots$

${e}^{i_{\ell-1}}_1 .. {e}^{i_{\ell-1}}_n$

$\bar{e}^{i_{\ell}}_1 \dots \bar{e}^{i_{\ell}}_n$

$\bar{e}^{i_{\ell+1}}_1 .. \bar{e}^{i_{\ell+1}}_n$

$\dots$

$\bar{e}^{i_r}_1 \dots \bar{e}^{i_r}_{n}$

$T_1 \dots T_{n}$
}
}

\noindent
guarantees that 
the maximal pair (in the new $2^{n-1}$ pairs) 
has the same eq-level as the 
original maximal
pair (in the originally assumed $2^n$ pairs),
and this eq-level is less than the eq-level of any other pair.
We can thus apply the induction hypothesis.
\qed
\end{proof}

\noindent
For stating an important corollary we introduce the following notion
(which assumes a fixed triple $\calG, T_0,U_0$).
Given a head pair $E(x_1,\dots,x_n),F(x_1,\dots,x_n)$
and head extensions
$(e^1_1,\dots,e^1_n)$, $\dots$, $(e^n_1,\dots,e^n_n)$,
we say that, for a tuple $T_1,\dots,T_n$,
the pair
$(E'_{max}(T_1,\dots,T_n),F'_{max}(T_1,\dots,T_n))$
is \emph{saturated on level} $m\in\Nat$ 
if for each other extended head pair $(E',F')$ there is $u$, 
$|u|<m$, such that 
$u\models (E'(T_1,\dots,T_n),F'(T_1,\dots,T_n))$.

\begin{cor}\label{cor:repeateqlevel}
If $u\models (U,U')$ for an offending prefix $u$ for $(T_0,U_0)$
where $(U,U')$ can be presented
as a pair $(E'_{max}(T_1,\dots,T_n),F'_{max}(T_1,\dots,T_n))$
saturated on level $|u|$
then 
there cannot exist
an offending prefix $v$, $|v|\neq |u|$,
and a tuple $T'_1,\dots,T'_n$
such that $v\models (V,V')$ where 
$(V,V')$ can be presented
as $(E'_{max}(T'_1,\dots,T'_n),F'_{max}(T'_1,\dots,T'_n))$
saturated on level $|v|$ (where the head pair $E,F$ and the head
extensions are the same in both cases).
\end{cor}

\noindent
We now fix $T_0,U_0$, where $T_0\not\sim U_0$,
and a (finite) $\alpha\in\offpl(T_0,U_0)$, and use  
the same (balancing) strategy 
along $\alpha$
as we used along the infinite
$\alpha$ in the proof of Proposition~\ref{prop:ghasstairbase}; i.e.,
we present $\alpha$ in the appropriate form $u_1v_1u_2v_2\dots$ as long as
possible. We note that we are guaranteed that $\alpha$ does not allow
a repeat, since it is offending;
if $u\models (T,U)$ for a prefix $u$ of $\alpha$
then $\eqlevel(T,U)=\eqlevel(T_0,U_0)-|u|$. 
Instead of an infinite pivot path 
we now get
a finite pivot path 
\begin{center}
$B_1\gt{w_1} B_2\gt{w_2}\cdots\gt{w_{k-1}}B_k$\,,
\end{center}
maybe with $k=0$ (i.e., with no balancing step at all).
But recall the \emph{pivot-path property}: each segment $V\gt{w}$
of length $|w|=M_2$ either contains a pivot, i.e.
$V\gt{u}B_i$ for a prefix $u$ of $w$, or sinks, i.e. $V$ sinks by $w$.

Each pivot $B_i$, for $i=1,2,\dots, k{-}1$,
has the unique subterm $V_i$ (maybe $V_i=B_i$)
which is \emph{exposed by a proper prefix of} 
$w_i$ but none of its 
proper subterms is thus exposed.
This yields the path 
\begin{equation}\label{eq:finitepivotpath}
B_1\gt{w_{11}}V_1\gt{w_{12}} B_2\gt{w_{21}}V_2
\gt{w_{22}}B_3\gt{w_{31}}
\cdots
\gt{w_{k-2,2}}B_{k-1}
\gt{w_{k-1,1}}V_{k-1}\gt{w_{k-1,2}} B_k
\end{equation}
where $w_{i1}$ can be empty but $w_{i2}$ are nonempty and 
 $|w_{i2}|\leq M_2$.
We note that 
for each segment $V_i\gt{w_{i2}} B_{i+1}\gt{w_{i+1,1}}V_{i+1}$
we either have that $V_i$ does not sink by 
${w_{i2}}{w_{i+1,1}}$ or $V_{i+1}$ is a subterm of $V_i$.

\begin{defn}
We call a subsequence $(i_0, i_1, \dots , i_r)$ of 
the sequence $(1,2,\dots,k{-}1)$ 
a \emph{stair sequence} 
if for each
$j\in\{0,1,\dots,r{-}1\}$
we have that 
$V_{i_j}$ does not sink by ${w}$ where 
$V_{i_j}\gt{w}V_{i_{j+1}}$ is the appropriate segment
of~(\ref{eq:finitepivotpath}), so $w=w_{(i_j,2)}w_{i_j+1}\dots
w_{i_{j+1}-1}w_{(i_{j+1},1)}$.
\\
A~\emph{stair sequence} is \emph{maximal} if 
it is not a proper subsequence of any other stair sequence.
\end{defn}

\begin{prop}\label{prop:stairpres}
If  $(i_0, i_1, \dots , i_r)$ is a maximal stair sequence then
$V_{i_0},V_{i_1},V_{i_2},\dots, V_{i_r}$ 
can be presented as

$V_{i_0}=$
\framebox{
\blaytree{14}{
$G_0$

$T_1  \dots T_n$
}
}
,
$V_{i_1}=$
\framebox{
\blaytree{14}{
$G_1$

$e^1_1  \dots e^1_n$

$T_1  \dots T_n$
}
}
,
$V_{i_2}=$
\framebox{
\blaytree{14}{
$G_2$

$e^2_1  \dots e^2_n$

$e^1_1  \dots e^1_n$

$T_1  \dots T_n$
}
}
,
$\dots$,
$V_{i_r}=$
\framebox{
\blaytree{14}{
$G_r$

$e^r_1  \dots e^r_n$

$\dots$

$e^2_1  \dots e^2_n$

$e^1_1  \dots e^1_n$

$T_1  \dots T_n$
}
}

\smallskip
\noindent
where $G_j$ are $M_0$-prefixes (so $n$ is simply bounded) and
$\depthsize(e^i_j)\leq \boundinc(M_2)$. 
\end{prop}

\begin{proof}
Assuming a maximal stair sequence,
we first show that $V_{i_{j+1}}$ is a subterm of $V_{i_{j}+1}$:
For any $\ell\geq 1$ such that  $i_{j}+\ell<i_{j+1}$ we must have 
that $V_{i_{j}+\ell}$ sinks by $w'$ for the appropriate
segment $V_{i_{j}+\ell}\gt{w'}V_{i_{j+1}}$, which implies that 
there is $\ell'$ such that 
$i_{j}+\ell<i_{j}+\ell'\leq i_{j+1}$ where
 $V_{i_{j}+\ell'}$ is  a subterm of $V_{i_{j}+\ell}$\,.

By definition, for the segment 
$V_{i_j}\gt{w}V_{i_{j+1}}$ we have $Yx_1\dots
x_m\gt{w}F(x_1,\dots,x_m)$ where $Y$ is the root nonterminal of
$V_{i_j}$ and $F$ is not a
variable. Recalling the above pivot-path property and the fact
that   $V_{i_{j+1}}$ is a subterm of $V_{i_{j}+1}$, we deduce that
$1\leq\depthsize(F)\leq 1+\boundinc(M_2)$.
This easily implies the claim.
(Note that $e^i_j$ can be just a variable. It is also possible
that some $T_j$, $e^i_j$ get obsolete, do not really matter in the
respective substitutions.)
\qed
\end{proof}

\begin{prop}
There is an elementary function $g$, independent of $\calG,
T_0,U_0,\alpha$, such that $|\alpha|\leq g(\inputsize,\ell)$
where $\ell$ is the length of the longest stair sequence. 
\end{prop}

\begin{proof}
For any ground term $V$ which can be presented as
$V=F(W_1,\dots,W_m)$ where $F$ is a finite term and  
 $W_i$ are subterms of $T_0$ or $U_0$, let us define
$size(V)$
as the least $\depthsize(F)$ in such presentations;
we note that if $V'$ is a subterm of $V$ then $size(V')\leq size(V)$.
Since either $T_0\gt{u_1}B_1$ or  $U_0\gt{u_1}B_1$, the above $size$ is
well defined for all $B_i$ and $V_i$\,;
moreover,  $size(V_{i+1})\leq  size(V_{i})+M_3$ where we put
$M_3=1+\boundinc(M_2)$.

We now show that 
if $size(V_{i})> p M_3$  (for $p\in\Nat$)
then there is a stair sequence $(i_0,i_1,\dots,i_p)$ such that
$i_p=i$\,:
Suppose $V_i$ is a counterexample for the least $i$ and some $p$\,;
we necessarily
have $p\geq 1$.
Since $V_1$ is a subterm of $B_1$ and $B_1$ is reachable from a
subterm of $T_0$ or $U_0$ by less than $M_0$ moves, we surely 
have $size(V_1)\leq M_3$; hence $i>1$.
We have  $size(V_{i-1})> (p{-}1)M_3$ and
$V_{i}$ is a subterm of $V_{i-1}$ (since $V_{i-1}$ satisfies the claim for
$p-1$ and thus $V_{i-1}\gt{w}V_i$ necessarily sinks); hence
$size(V_{i-1})> p M_3$.  Then $i{-}1>1$ and 
we also have  $size(V_{i-2})> (p{-}1)M_3$ and $V_i$ is a
subterm of $V_{i-2}$, so $size(V_{i-2})> p M_3$; 
continuing this reasoning leads to a contradiction.

Hence if $\ell$  is the length of the longest stair sequence 
then we have $size(B_i)\leq (\ell+2)M_3$ for
each pivot $B_i$, which gives an elementary bound (in $\inputsize$ and
$\ell$) on the length of the pivot path, and thus an elementary bound
on $|\alpha|$ as well.
\qed
\end{proof}

\noindent
Thus to prove Theorem~\ref{th:lengthbound}, it is sufficient to show
Proposition~\ref{prop:boundstairsequence}.
We show this by using Corollary~\ref{cor:repeateqlevel} and
a few combinatorial facts. 
(The combinatorial facts could be surely found in the literature in
some form, e.g. for so called \emph{Zimin words},
but we provide 
short self-contained versions tailored to our aims.)

\begin{defn}\label{def:typewords}
Given an 
alphabet $\Sigma$, the empty word 
$\varepsilon$ is of \emph{type $0$};
for $n\geq 1$, a word $w\in\Sigma^*$ is of \emph{type $n$} if $w=vuv$
for some $v$ of type $n{-}1$ and some $u$, $|u|\geq 1$.
Each word $w$ of type $n>0$ has 
a \emph{type-$n$ presentation} given by
nonempty words $v_1,v_2,\dots,v_n$; these give rise to words
$w_1,w_2,\dots,w_n$, where $w_i$ is of type $i$, as follows: 
$w_1=v_1$, $w_2=w_1v_2w_1$ ($=v_1v_2v_1$),
$w_3=w_2v_3w_2$ ($=v_1v_2v_1v_3v_1v_2v_1$), $\dots$,
$w_n=w_{n-1}v_nw_{n-1}$, where $w_n=w$.
\end{defn}

\begin{prop}\label{prop:prefixyielding}
For a type-$(n{+}1)$ presentation of $w\in\Sigma^*$,
given by words $v_1,v_2,\dots,v_{n+1}$, there are words 
$u_1,u_2,\dots,u_n$, all beginning 
with the first symbol of $v_1$,
such that 
$u_{i_1}u_{i_2}\dots u_{i_r}$ is a suffix of
(the right quotient)
$w/v_1$
for all $1\leq r\leq n$ and $1\leq i_1<i_2<\cdots < i_r\leq n$.
\end{prop}

\begin{proof}
The words $v_1,v_2,\dots,v_{n+1}$ give rise to $w_1,w_2,\dots,w_{n+1}$ as
in Definition~\ref{def:typewords}. We  denote $w'_i=w_i/v_1$, 
and put $u_i=v_1v_{i+1}w'_i$ for $i=1,2,\dots,n$
(so $w'_{i+1}=w'_iu_i=w'_iv_1v_{i+1}w'_i$). 
Hence

\begin{minipage}{0.3\textwidth}
\begin{tabbing}
$u_3=$\=$v_1v_4$\=$v_1v_2$\=$v_1v_3$\=$v_1$\=$v_2$\kill
$w'_1=$\>\>\>\>\>$\varepsilon$\\
$w'_2=$\>\>\>\>$v_1v_2$\\
$w'_3=$\>\>$v_1v_2v_1v_3v_1v_2$\\
$\dots$
\end{tabbing}
\end{minipage}
\hspace{1em}
\begin{minipage}{0.3\textwidth}
\begin{tabbing}
$u_3=$\=$v_1v_4v_1v_2$\=$v_1v_3$\=$v_1v_2$\kill
$u_1=$\>\>\>$v_1v_2$\\
$u_2=$\>\>$v_1v_3v_1v_2$\\
$u_3=$\>$v_1v_4v_1v_2v_1v_3v_1v_2$\\
$\dots$
\end{tabbing}
\end{minipage}

\noindent
Since $w'_i$ is a suffix of $w'_j$ for each $j>i$, it is sufficient to
show that 
$u_{i_1}u_{i_2}\dots u_{i_r}$ is a suffix of
$w'_{i_r+1}$.
Since $w'_{i_r+1}=w'_{i_r}u_{i_r}$ and by induction hypothesis we can
assume that
$u_{i_1}u_{i_2}\dots u_{i_{r-1}}$ 
is a suffix of $w'_{i_{r-1}+1}$, and thus of $w'_{i_r}$, the claim is
clear.
\qed
\end{proof}

\noindent
For $h\in\Nat$, we define  function $f_h:\Nat\rightarrow\Nat$
(recursively) 
and note the next proposition:
\[ f_h(0)=0,\ \  f_h(n+1)=(1+f_h(n))\cdot h^{f_h(n)}\,. \]

\begin{prop}\label{prop:longword}
If $|\Sigma|=h$ then 
each $w\in \Sigma^*$, $|w|\geq f_h(n)$, 
contains a subword of type $n$.
\end{prop}

\begin{proof}
By induction, using the pigeonhole
principle. For $n=0$ the claim is obvious.
Any word $w$ of
length $f_h(n+1)$ necessarily 
contains two occurrences of some
(sub)word $u$ of length $f_h(n)$ separated by a nonempty word. 
Thus $w=v_1uv_2uv_3$ where  $|v_2|\geq 1$ and 
$u=u_1u_2u_3$ with  $u_2$ of type $n$; 
this means that the subword
$u_2u_3v_2u_1u_2$ is of type $n+1$.
\qed
\end{proof}

\begin{prop}\label{prop:boundstairsequence}
There is an elementary function $g$,
independent of $\calG, T_0, U_0, \alpha$,
such that the length of any stair sequence 
has an upper bound 
$2\uparrow\uparrow g(size(\calG))$. 
\end{prop}

\begin{proof} 
Given a maximal stair sequence $(i_0, i_1, \dots , i_r)$
and the presentation of $V_{i_0},V_{i_1},V_{i_2},\dots, V_{i_r}$ as in
Proposition~\ref{prop:stairpres},
let us
consider the pivot (sub)sequence $B_{i_0+1},\dots, B_{i_{r-1}+1}$
($B_{i_j+1}$ is the first pivot after $V_{i_{j}}$).
The balancing results with 
$B_{i_0+1},\dots, B_{i_{r-1}+1}$
can be presented as 

\framebox{
\blaytree{13}{
$E_0$

$T_1  \dots T_n$
}
}
\framebox{
\blaytree{13}{
$F_0$

$T_1  \dots T_n$
}
}
,
\framebox{
\blaytree{13}{
$E_1$

$e^1_1 \dots e^1_n$

$T_1  \dots T_n$
}
}
\framebox{
\blaytree{13}{
$F_1$

$e^1_1 \dots e^1_n$

$T_1  \dots T_n$
}
}
,$\dots$,
\framebox{
\blaytree{17}{
$E_{r-1}$

$e^{r-1}_1 \dots e^{r-1}_n$

$\dots$

$e^2_1 \dots e^2_n$

$e^1_1 \dots e^1_n$

$T_1  \dots T_n$
}
}
\framebox{
\blaytree{17}{
$F_{r-1}$

$e^{r-1}_1 \dots e^{r-1}_n$

$\dots$

$e^2_1 \dots e^2_n$

$e^1_1 \dots e^1_n$

$T_1  \dots T_n$
}
}

\smallskip
\noindent
where $E_i,F_i,e^i_1 \dots e^i_n$ are finite terms
whose depth-size is simply bounded 
(unlike
in/around Proposition~\ref{prop:afixedpair}
where we considered general regular terms).
Hence 
the tuples $(E_i,F_i,e^i_1, \dots, e^i_n)$ can be viewed as
elements of an alphabet with $h$ elements, where 
$h$ is bounded by an elementary function of $size(\calG)$.
If $r{-}1\geq f_h(n{+}2)$ then the word
\begin{center}
$(E_{r-1},F_{r-1},e^{r-1}_1,\dots,e^{r-1}_n)$
$(E_{r-2},F_{r-2},e^{r-2}_1,\dots,e^{r-2}_n)$
$\dots$
$(E_1,F_1,e^1_1,\dots,e^1_n)$
\end{center}
contains two different occurrences  
of a subword of type $n{+}1$, 
by Proposition~\ref{prop:longword}.
By Proposition~\ref{prop:prefixyielding},
from this subword we can extract a pair $E,F$ (i.e., the
``head-projection'' of the first symbol of $v_1$),
and  $n$ words (``extension-projections'' of $u_1,\dots,u_n$)
$(\bar{e}^1_1, \dots, \bar{e}^1_n)$,
 $(\bar{e}^2_1, \dots, \bar{e}^2_n)$, $\dots$,
  $(\bar{e}^n_1, \dots, \bar{e}^n_n)$, where each 
$(\bar{e}^j_1, \dots, \bar{e}^j_n)$ corresponds to the substitution
arising by composing the substitutions corresponding to a segment
\framebox{
\blaytree{18}{

$e^{k+\ell}_1 \dots e^{k+\ell}_n$

$\dots$

$e^k_1 \dots e^k_n$
}
}
so that:
there are $U_1,U_2\dots,U_n$ determined by the first occurrence of
the type $n{+}1$ subword
such that
for each tuple $1\leq i_1<i_2<\cdots < i_r\leq n$ there 
is a pair 
\framebox{
\blaytree{13}{
$E$

$\bar{e}^{i_1}_1 \dots \bar{e}^{i_1}_n$

$\dots$

$\bar{e}^{i_r}_1 \dots \bar{e}^{i_r}_n$

$U_1  \dots U_n$
}
}
\framebox{
\blaytree{13}{
$F$

$\bar{e}^{i_1}_1 \dots \bar{e}^{i_1}_n$

$\dots$

$\bar{e}^{i_r}_1 \dots \bar{e}^{i_r}_n$

$U_1  \dots U_n$
}
}
in the above sequence, where the appropriate 
$(E'_{max}(U_1,\dots,U_n),F'_{max}(U_1,\dots,U_n))$
is saturated on the corresponding level. Similarly for
$V_1,V_2,\dots,V_n$ 
determined by the second occurrence of the type $n{+}1$ subword.
This 
would yield a contradiction
with Corollary~\ref{cor:repeateqlevel}.

Therefore $r-1<f_h(n+2)$. 
From the definition of $f_h$ we can easily derive 
$f_h(n+2)\leq h \uparrow\uparrow g_1(n)$ for an elementary function 
$g_1$.
Since $h$ and $n$ are bounded by elementary
functions of $size(\calG)$,
the claim follows.
\qed
\end{proof}

\end{document}